\documentclass[a4paper,reqno,twoside]{amsart}

\usepackage[left=3.5cm,right=3.5cm,top=3.5cm,bottom=3.5cm,footskip=1.5cm,headsep=1cm,bindingoffset=0.cm,]{geometry}

\usepackage[utf8]{inputenc}
\usepackage[english]{babel}
\usepackage[T1]{fontenc} % To switch to the T1 encoding
\usepackage{lmodern} % To switch to Latin Modern
\rmfamily % To load Latin Modern Roman and enable the following NFSS declarations.
\DeclareFontShape{T1}{lmr}{b}{sc}{<->ssub*cmr/bx/sc}{}
\DeclareFontShape{T1}{lmr}{bx}{sc}{<->ssub*cmr/bx/sc}{}

\numberwithin{equation}{section}
\usepackage{amsmath}
\usepackage{amssymb}
\usepackage{amsthm}
\usepackage{amsfonts}
\usepackage{mathtools}

\usepackage{mathdots}
\usepackage{caption}
\usepackage{subcaption}

\usepackage{dsfont} % for 1 bbmath
\usepackage[format=hang]{caption}
\usepackage[normalem]{ulem}

\usepackage{standalone}
\usepackage{graphicx}
\usepackage{imakeidx}
\makeindex
\usepackage{setspace}
\usepackage{appendix}
\usepackage{pdfpages}
\usepackage{upgreek}
\usepackage{tabulary}
\usepackage{booktabs}
\usepackage{extarrows}
\usepackage{tabularx}
\usepackage{float}
\restylefloat{table}
\usepackage{enumitem}
\usepackage{csquotes}
\usepackage{bm}
\usepackage{orcidlink}
\usepackage{lipsum}
\usepackage{bbold}

\usepackage{algorithm}
\usepackage{algpseudocode}
\algrenewcommand{\algorithmicrequire}{\textbf{Input:}}
\algrenewcommand{\algorithmicensure}{\textbf{Output:}}

\usepackage{tikz}
\usepackage{tikz-layers}
\usepackage{tikz-cd}
\usepackage[arrow, matrix, curve]{xy}
\usetikzlibrary{matrix,positioning,decorations.pathreplacing,calc}
\usetikzlibrary{
    decorations.text,%
    decorations.markings,%
    shadows}
\usetikzlibrary{calc,intersections}
\usepackage{adjustbox}
\usepackage{graphicx}

\usepackage[
    style=numeric,
    sorting=nty,
    maxnames=99,
    maxalphanames=5,
    natbib=true,
    backend=bibtex,
    sortcites]{biblatex}

\DeclareNameAlias{default}{family-given}

\AtEveryBibitem{% Clean up the bibtex rather than editing it
    \clearfield{url}
    \clearfield{issn}
    \clearfield{isbn}
    \clearfield{urldate}

    \ifentrytype{book}{
        \clearfield{pages}}{% 
    }
}

\usepackage{xargs}                      % Use more than one optional parameter in a new commands
\usepackage[prependcaption,textsize=tiny,textwidth=3cm]{todonotes}
\newcommandx{\unsure}[2][1=]{\todo[linecolor=red,backgroundcolor=red!25,bordercolor=red,#1]{#2}}
\newcommandx{\change}[2][1=]{\todo[linecolor=blue,backgroundcolor=blue!25,bordercolor=blue,#1]{#2}}
\newcommandx{\info}[2][1=]{\todo[linecolor=OliveGreen,backgroundcolor=OliveGreen!25,bordercolor=OliveGreen,#1]{#2}}
\newcommandx{\improvement}[2][1=]{\todo[linecolor=black,backgroundcolor=black!25,bordercolor=black,#1]{#2}}
\newcommandx{\thiswillnotshow}[2][1=]{\todo[disable,#1]{#2}}

\usepackage[noabbrev,capitalize]{cleveref}
\crefname{proposition}{Proposition}{Propositions}
\crefname{equation}{}{}

\newtheorem{theorem}{Theorem}[section]
\newtheorem{lemma}[theorem]{Lemma}
\newtheorem{proposition}[theorem]{Proposition}

\theoremstyle{definition}
\newtheorem{definition}[theorem]{Definition}
\newtheorem{example}[theorem]{Example}

\newtheorem{remark}[theorem]{Remark}

\crefname{assumption}{Assumption}{Assumptions}
\crefname{definition}{Definition}{Definitions}
\crefname{corollary}{Corollary}{Corollaries}
\crefname{enumi}{item}{items}

\creflabelformat{subfigure}{#2\textsc{#1}#3}

\usepackage[final]{microtype}

\makeatletter
\newsavebox\myboxA
\newsavebox\myboxB
\newlength\mylenA

\newcommand*\xoverline[2][0.75]{%
  \sbox{\myboxA}{$\m@th#2$}%
  \setbox\myboxB\null% Phantom box
  \ht\myboxB=\ht\myboxA%
  \dp\myboxB=\dp\myboxA%
  \wd\myboxB=#1\wd\myboxA% Scale phantom
  \sbox\myboxB{$\m@th\overline{\copy\myboxB}$}%  Overlined phantom
  \setlength\mylenA{\the\wd\myboxA}%   calc width diff
  \addtolength\mylenA{-\the\wd\myboxB}%
  \ifdim\wd\myboxB<\wd\myboxA%
    \rlap{\hskip 0.5\mylenA\usebox\myboxB}{\usebox\myboxA}%
  \else
    \hskip -0.5\mylenA\rlap{\usebox\myboxA}{\hskip 0.5\mylenA\usebox\myboxB}%
  \fi}
\makeatother

\usepackage{fancyhdr}

\fancypagestyle{plain}{%
    \fancyhead[C]{}
    \fancyfoot[C]{}
    
    }
\fancypagestyle{nosection}{%
    \fancyhead[CE]{}
    \fancyhead[CO]{}
    \fancyfoot[CE,CO]{\thepage}
}

\usepackage{tikz}

\usetikzlibrary{matrix} 
\usetikzlibrary{arrows,arrows.meta,bending} %new code

\usepackage{nicematrix}

\usetikzlibrary{hobby}

\usepackage{calc} 

\newcommand{\lz}{\ell_2(\Z)}
\newcommand{\ldz}{\mathbb{L}_D(\Z)}
\newcommand{\ldm}{\mathbb{L}_D(M)}

\newcommand{\qq}[1]{\bm q^{(#1)}}
\newcommand{\sss}[1]{\bm s^{(#1)}}
\newcommand{\uu}[1]{\bm u^{(#1)}}

\newcommand{\oo}[1]{\chi^{(#1)}}
\newcommand{\mm}[1]{{\mu^{(#1)}}}

\DeclareMathOperator{\prob}{P}
\DeclareMathOperator{\eCDF}{eCDF}
\DeclareMathOperator{\len}{len}

\DeclareMathOperator{\N}{\mathbb{N}}
\DeclareMathOperator{\Z}{\mathbb{Z}}

\DeclareMathOperator{\C}{\mathbb{C}}
\DeclareMathOperator{\Var}{Var}

\DeclareMathOperator{\tr}{tr}
\renewcommand{\i}{\mathbf{i}}

\renewcommand{\hat}{\widehat}

\renewcommand{\qedsymbol}{$\blacksquare$}

\newcommand{\ds}{\displaystyle}

\DeclareMathOperator{\diag}{diag}
\DeclareMathOperator{\BO}{\mathcal{O}}

\renewcommand{\epsilon}{\varepsilon}
\DeclareMathOperator{\dd}{d\!}

\renewcommand{\i}{\mathbf{i}}

\renewcommand{\hat}{\widehat}

\DeclareMathOperator{\iL}{{\mathsf{L}}}
\DeclareMathOperator{\iR}{{\mathsf{R}}}
\DeclareMathOperator{\iLR}{{\mathsf{L},\mathsf{R}}}

\usepackage{tcolorbox}
\usetikzlibrary{tikzmark,calc}

\newcommand{\ip}[2]{\left\langle #1, #2 \right\rangle}
\newcommand{\mc}[1]{\mathcal{#1}}

\newcommand{\abs}[1]{\left\lvert#1\right\rvert}

\newcommand{\norm}[1]{\left\lVert#1\right\rVert}

\newcommandx{\silvio}[2][1=]{\todo[linecolor=blue,backgroundcolor=blue!25,bordercolor=blue,#1]{Silvio: #2}}

\newcommandx{\alex}[2][1=]{\todo[linecolor=red,backgroundcolor=red!25,bordercolor=red,#1]{Alex: #2}}

\title[Density of states for random block systems]{Universal estimates for the density of states for aperiodic block subwavelength resonator systems}

\begin{document}
 \author[H. Ammari]{Habib Ammari \,\orcidlink{0000-0001-7278-4877}}
 \address{\parbox{\linewidth}{Habib Ammari\\
  ETH Z\"urich, Department of Mathematics, Rämistrasse 101, 8092 Z\"urich, Switzerland, \href{http://orcid.org/0000-0001-7278-4877}{orcid.org/0000-0001-7278-4877}}.}
 \email{habib.ammari@math.ethz.ch}
 \thanks{}

 \author[S. Barandun]{Silvio Barandun\,\orcidlink{0000-0003-1499-4352}}
  \address{\parbox{\linewidth}{Silvio Barandun\\
  ETH Z\"urich, Department of Mathematics, Rämistrasse 101, 8092 Z\"urich, Switzerland, \href{http://orcid.org/0000-0003-1499-4352}{orcid.org/0000-0003-1499-4352}}.}
  \email{silvio.barandun@sam.math.ethz.ch}

\author[B. Davies]{Bryn Davies\,\orcidlink{0000-0001-8108-2316}}
 \address{\parbox{\linewidth}{Bryn Davies\\
Mathematics Institute, University of Warwick, Zeeman Building, Coventry CV4 7AL, UK,
 \href{http://orcid.org/0000-0001-8108-2316}{orcid.org/0000-0001-8108-2316}
}}
\email{bryn.davies@warwick.ac.uk}

 \author[E.O. Hiltunen]{Erik Orvehed Hiltunen\,\orcidlink{0000-0003-2891-9396}}
\address{\parbox{\linewidth}{Erik Orvehed Hiltunen\\
Department of Mathematics, University of Oslo, Moltke Moes vei 35, 0851 Oslo, Norway, 
 \href{http://orcid.org/0000-0003-2891-9396}{orcid.org/0000-0003-2891-9396}}}
\email{erikhilt@math.uio.no}

 \author[A. Uhlmann]{Alexander Uhlmann\,\orcidlink{0009-0002-0426-6407}}
  \address{\parbox{\linewidth}{Alexander Uhlmann\\
  ETH Z\"urich, Department of Mathematics, Rämistrasse 101, 8092 Z\"urich, Switzerland, \href{http://orcid.org/0009-0002-0426-6407}{orcid.org/0009-0002-0426-6407}}.}

 \email{alexander.uhlmann@sam.math.ethz.ch}

 \begin{abstract}
We consider the spectral properties of aperiodic block subwavelength resonator systems in one dimension, with a primary focus on the density of states. We prove that for random block configurations, as the number of blocks $M\to \infty$, the integrated density of states converges to a non-random, continuous function. We show both analytically and numerically that the density of states exhibits a tripartite decomposition: it vanishes identically within bandgaps; it forms smooth, band-like distributions in shared pass bands (a consequence of constructive eigenmode interactions); and, most notably, it exhibits a distinct fractal-like character in hybridisation regions. We demonstrate that this fractal-like behaviour stems from the limited interaction between eigenmodes within these hybridisation regions. Capitalising on this insight, we introduce an efficient meta-atom approach that enables rapid and accurate prediction of the density of states in these hybridisation regions. This approach is shown to extend to systems with quasiperiodic and hyperuniform arrangements of blocks.
 \end{abstract}

\maketitle

\par

\bigskip

\\

\section{Introduction}

In this paper, we study the spectral properties of large one-dimensional resonator systems modelled by the Helmholtz equation \eqref{waveeq}, where we assume that the contrast between the material parameters inside and outside the resonators, $\rho_b/\rho$, is small and the wave speeds $v$ and $v_i$ for $i=1,\ldots, N,$ are of order one with $N$ being the number of resonators. In this setting, it was shown in \cite{feppon.cheng.ea2023Subwavelength} that the system possesses $N$ subwavelength resonant frequencies, \emph{i.e.}, with corresponding eigenmodes having wavelengths much larger than the typical size of the resonators or, equivalently, being almost constant inside the resonators. In contrast with higher frequency resonant modes, known as Fabry--Perot eigenmodes \cite{fabry-perot}, the collective behaviour of subwavelength eigenmodes can be used to design metamaterials that are microstructured materials with unusual wave scattering and transport properties at operating wavelengths comparable to their total size \cite{ammari.davies.ea2021Functional,review2,review1,lemoult.fink.ea2011Acoustic,lemoult.kaina.ea2016Soda,yves.fleury.ea2017Crystalline}. Furthermore, the subwavelength resonant frequencies and eigenmodes can be approximated using the 
generalised capacitance matrix given by $VC$, where $V$ and $C$ are defined by \eqref{def:v} and \eqref{def:C}, respectively. We refer to the resonance problem for the system of Helmholtz equations \eqref{waveeq} as the continuous model and to the eigenvalue problem for the generalised capacitance matrix as the discrete model.  

We consider large aperiodic arrays of such subwavelength resonators, constructed by sampling randomly or quasiperiodically from a fixed set of \enquote{blocks}. As first shown in \cite{disorder}, such systems have quite unusual spectral behaviour: existence of hybridised bound eigenmodes, two localisation mechanisms: bandgap localisation and disorder localisation, localisation of all the eigenmodes as the size of the structure goes to infinity, and fractal behaviour of eigenmodes. Here, we provide universal estimates for the density of states (DoS) for aperiodic block subwavelength resonator systems.  In systems of subwavelength resonators, the density of states measures, in some sense, how many eigenmodes correspond to eigenfrequencies below an operating low frequency $\omega$.  
We prove that the DoS of random block disordered systems is non-random and continuous (\emph{a.e.}) as the number of blocks goes to infinity. We distinguish three different regions for the DoS: the bandgap region where the DoS is zero, the shared pass band where the DoS is a smooth band due to ``constructive interactions'', and the hybridisation region where the DoS shows fractal-like behaviour. This fractal-like DoS stems from the limited interaction between the eigenmodes in the hybridisation region. This observation allows us to quickly and accurately predict the DoS in this region using a meta-atom approach. Our approach exploits the local arrangements of block and also works for quasiperiodic and hyperuniform block disordered systems. Note that all the specific examples of block subwavelength resonator systems studied in this paper can be grouped into a more general class, that is, the class of Markov chain block systems. The theory of Jacobi operators and matrices with entries obeying a Markov process can be applied to describe their spectral properties.
 
It is worth mentioning that our estimates on the DoS in this paper extend those obtained on electronic systems described by discrete random Schr\"odinger operators (see, for example, \cite{damanik_book1, damanik_book2,bary_book} and the references therein), in particular the random word models \cite{damanik}, to classical systems of subwavelength resonators where, despite being in the low-frequency regime, the subwavelength resonances of the system interact strongly with the disorder. It should be noticed that the localisation mechanism for the Schr\"odinger models, which is due to potential wells, is fundamentally different from the one for the discrete wave models studied here. Furthermore, in electronic models, one perturbs the potential, which induces perturbations at the discrete level only of the diagonal terms of the corresponding Hamiltonian matrices, while for systems of subwavelength resonators, physical perturbations affect all the non-zero entries of the corresponding discrete approximations that are given in terms of the generalised capacitance matrices.

The tools used for studying localisation properties in aperiodic block subwavelength resonator systems in one dimension and deriving universal estimates for the DoS are the transfer and propagation matrices, the Thouless criterion of localisation, and the theory of Jacobi operators and matrices. The concept of metrically transitive operators is needed to prove ergodicity of spectral properties. 

While some approaches and results, such as the bandgap localisation mechanism and the periodisation technique for detecting the bandgaps and the mobility edges, are expected to hold in three dimensions, most of the results obtained in this paper as well as in \cite{disorder} (in particular, the fractal behaviour of the density of states in the hybridised region and the localisation of all the eigenmodes as the size of the system goes to infinity) are specific to the one-dimensional case. The extension of the tools used in this paper for the study of spectral properties of disordered subwavelength resonator systems to three dimensions seems to be difficult even though the generalised capacitance matrix can be approximated by banded Jacobi matrices. Note also that the non-Hermitian skin effect \cite{jahan,zhang.zhang.ea2022review,okuma.kawabata.ea2020Topological,ammari.barandun.ea2023NonHermitian,ammari.barandun.ea2024Mathematical}
in disordered resonator systems of subwavelength scales could be studied similarly using the theory of non-Hermitian Jacobi matrices and operators. Both the extensions of our results and methods to the three-dimensional case and to the non-reciprocal setting of the non-Hermitian skin effect are important and challenging problems. They will be the subject of forthcoming works. 

The paper is organised as follows. In \cref{sec:setting}, we introduce the notion of finite \emph{block disordered systems} of subwavelength resonators and recall the capacitance matrix and the propagation matrix formalisms to study their spectral properties.  In \cref{sec:jacobi}, we use the notion of \emph{metrically transitive operators} from \cite{pastur1992Spectra} to obtain a variety of results regarding the spectra of random block disordered systems, in particular,  ergodicity and spectral convergence. In \cref{sec:metaatom}, we introduce an efficient way to calculate the DoS for aperiodic block systems and estimate the DoS in the hybridised region. In \cref{sec:dependent sampling}, we investigate the DoS in the hybridisation regions of block disordered systems when the block sequence is no longer sampled independently. We consider multiple alternative sampling, including hyperuniform sampling and quasiperiodic sequences, and see that in all of these cases the meta-atom approximation continues to function---often even better than in the \emph{i.i.d.} sampling case. \cref{sec:hyper-demo} is devoted to hyperuniform sampling. Finally, in \cref{sec:thouless}, we recall the Thouless criterion developed in \cite{disorder} to study localisation in disordered subwavelength resonator systems.  

\section{Setting and tools}\label{sec:setting}
We begin by introducing the notions of \emph{block disordered} systems of \emph{subwavelength resonators} and the main tools used to study them, including the capacitance matrix approximation and the propagation matrix formalism. 
To that end, this section will contain some results from \cite{disorder}, and we refer to this work for further details.

\subsection{Subwavelength resonators}\label{sec:subwavelength setting}
The central systems of interest in this work are one-dimensional chains of subwavelength resonators (see \cite{ammari.davies.ea2021Functional,feppon.cheng.ea2023Subwavelength, cbms,disorder}). That is, we consider an array $\mc D=\bigsqcup_{i=1}^N (x_i^{\iL}, x_i^{\iR})$ consisting of $N$ resonators $D_i = (x_i^{\iL}, x_i^{\iR})$. We will denote $\ell_i = x_i^{\iR} - x_i^{\iL}$ for $1\leq i \leq N$ the \emph{length} of the $i$\textsuperscript{th} resonator and $s_i = x_{i+1}^{\iL}-x_{i}^{\iR}$ for $1\leq i \leq N-1$ the \emph{spacing} between the $i$\textsuperscript{th} and $(i+1)$\textsuperscript{th} resonator. 

The resonators are distinct from the background medium due to differing wave speeds and densities, given by
\begin{equation}
    v(x) = \begin{cases}
        v_i &\text{if } x\in D_i,\\
        v &\text{if } x\in \mathbb{R}\setminus \mc D, 
    \end{cases} \quad
    \rho(x) = \begin{cases}
        \rho_b &\text{if } x\in D_i,\\
        \rho &\text{if } x\in \mathbb{R}\setminus \mc D.
    \end{cases}
\end{equation}
After further imposing an outward radiation condition, we obtain the following coupled system of Helmholtz equations for the resonant modes (see \cite[(1.6)]{feppon.ammari2022Subwavelength}). The resonance problem is to find $\omega$ such that there is a non-trivial solution $u$ to

\begin{align}
\label{waveeq}
    \begin{dcases}
\ds \frac{\mathrm{d}^2}{\mathrm{d}x^2} u + \frac{\omega^2}{v^2} u = 0,  &\text{in }  
\mathbb{R}\setminus \mc D,\\
\ds \frac{\mathrm{d}^2}{\mathrm{d}x^2} u + \frac{\omega^2}{v_i^2} u = 0,  & \text{in }   D_i, i=1,\ldots, N,\\
 u\vert_{\iR}(x^{\iLR}_i) - u\vert_{\iL}(x^{\iLR}_i) = 0,                                                                 & \text{for } i=1, \ldots, N,\\
        \left.\frac{\dd u}{\dd x}\right\vert_{\iR}(x^{\iL}_i) - \frac{\rho_b}{\rho} \left.\frac{\dd u}{\dd x}\right\vert_{\iL}(x^{\iL}_i) = 0, &   \text{for } i=1, \ldots, N,          \\
        \left.\frac{\dd u}{\dd x}\right\vert_{\iL}(x^{\iR}_i) - \frac{\rho_b}{\rho} \left.\frac{\dd u}{\dd x}\right\vert_{\iR}(x^{\iR}_i) = 0, &   \text{for } i=1, \ldots, N,          \\
\bigg( \frac{\mathrm{d}}{\mathrm{d} |x|} - \i \frac{\omega}{v} \bigg) u(x) =0, & \text{for } |x| \text{ large enough,} 
\end{dcases}
\end{align}
where for a function $w$ we denote by 
\begin{align*}
    w\vert_{\iL}(x) \coloneqq \lim_{s \downarrow 0} w(x-s) \quad \mbox{and} \quad  w\vert_{\iR}(x) \coloneqq \lim_{s \downarrow 0} w(x+s)
\end{align*}
if the limits exist. 

We are interested in the \emph{subwavelength high-contrast regime}. Namely, we denote by $\delta = \rho_b / \rho$ the \emph{contrast parameter} and consider the resonant frequencies  $\omega(\delta)$ with non-negative real part that satisfy
\begin{equation*}
    \omega(\delta) \to 0 \quad \text{ as } \quad \delta \to 0.
\end{equation*}

As shown in \cite{feppon.cheng.ea2023Subwavelength}, in this regime there exist exactly $N$ subwavelength resonant frequencies, characterised in leading order by a \emph{material matrix} $V$ that is diagonal and a tridiagonal \emph{capacitance matrix} $C$, respectively defined as
\begin{gather} 
    V = \diag \left(\frac{v_1^2}{\ell_1}, \dots, \frac{v_N^2}{\ell_N}\right) \in \mathbb{R}^{N\times N}, \label{def:v}
 \end{gather}
 and    
    \begin{gather}
    C = \left(\begin{array}{cccccc}
         \frac{1}{s_1}& -\frac{1}{s_1} \\
         -\frac{1}{s_1}& \frac{1}{s_1}+\frac{1}{s_2}& -\frac{1}{s_2} \\
         & -\frac{1}{s_2} & \frac{1}{s_2}+\frac{1}{s_3}& -\frac{1}{s_3}\\
         &&\ddots&\ddots&\ddots \\
         &&&-\frac{1}{s_{N-2}}& \frac{1}{s_{N-2}}+\frac{1}{s_{N-1}}& -\frac{1}{s_{N-1}}\\
         &&&&-\frac{1}{s_{N-1}}&\frac{1}{s_{N-1}}
    \end{array}\right) \in \mathbb{R}^{N\times N}. \label{def:C}
\end{gather}
The following results are from \cite{feppon.cheng.ea2023Subwavelength}. 
\begin{theorem}\label{thm:capapprox}
    Consider a system consisting of $N$ subwavelength resonators in $\mathbb{R}$. Then, there exist exactly $N$ subwavelength resonant frequencies $\omega(\delta)$ that satisfy $\omega(\delta)\to 0$ as $\delta\to 0$. Moreover, the $N$ resonant frequencies are given by 
    \[
        \omega_i(\delta) = \sqrt{\delta\lambda_i} + \mc O (\delta),
    \]
    where $(\lambda_i)_{1\leq i\leq N}$ are the eigenvalues of the eigenvalue problem
    \[
        VC\bm u_i = \lambda_i \bm u_i.
    \]
    Furthermore, the corresponding resonant modes $u_i(x)$ are approximately given by 
    \[
        u_i(x) = \sum_{j=1}^N \bm u_i^{(j)}V_j(x) + \mc O (\delta),
    \]
    where $\bm u_i^{(j)}$ is the $j$\textsuperscript{th} entry of the vector $\bm u_i$ and $V_j(x)$, $j=1,\ldots, N$, are defined as the solution to
    \begin{align*}
        \begin{dcases}
          -\frac{\dd{^2}}{\dd x^2} V_j(x) =0, & x\in\mathbb{R}\setminus \mc D, \\
          V_j(x)=\delta_{ij},              & x\in D_i, \ i=1,\ldots, N,                         \\
          V_j(x) = \BO(1),                  & \mbox{as } \abs{x}\to\infty.
        \end{dcases}
        \label{eq: def V_i}
    \end{align*}
\end{theorem}

We can thus fully understand the subwavelength resonant modes by studying the eigenvalue problem for the \emph{generalised capacitance matrix} $\mc C \coloneqq VC$ and will often use $\lambda_i$ and $\omega_i$ interchangeably, where it is clear from context. Note that the eigenvalue problem for $\mc C$ plays an analogous role as the tight binding model for Schr\"odinger operators. Therefore, we refer to it as the discrete model for the Helmholtz operator in \eqref{waveeq}.

\subsection{Block disordered systems}
Since we are interested in the band structure of disordered systems, where local translation invariance is broken, we must develop ways to construct and describe such systems. As discussed in \cite{disorder}, one such way is to consider a finite number of distinct \enquote{building blocks} consisting of (possibly multiple) subwavelength resonators. Later, constructing disordered resonator arrays from simple building blocks will enable us to characterise the ``unusual'' spectral properties of the array by looking at the building blocks.

A subwavelength block disordered system is a chain of subwavelength resonators consisting of $M$ blocks of resonators $B_{\oo{j}}$ sampled accordingly to a sequence $$\chi \in \{1,\dots , D\}^M \eqqcolon \ldm$$ from a selection $B_1, \dots, B_D$ of $D$ distinct resonator blocks, arranged on a line. Each resonator block $B_j$ is characterised by a sequence of tuples $(v_1,\ell_1,s_1), \dots, (v_{\len(B_j)},\ell_{\len(B_j)},s_{\len(B_j)})$ that denote the wave speed, length, and spacing of each constituent resonator. Here, $\len(B_j)$ denotes the total number of resonators contained within the block $B_j$.

We will often abuse notation and write $\mc D = \bigcup_{j=1}^M B_{\oo{j}}$ to denote the resonator array constructed by sampling the blocks $B_1, \dots B_D$  according to $\chi\in \ldm$ and then arranging them in a line. Having thus constructed an array of subwavelength resonators, we can alternatively write that $\mc D = \bigcup_{i=1}^N D_i$ in line with \cref{sec:subwavelength setting}.
Note that as $M$ denotes the total number of sampled blocks and $N$ the total number of resonators, we always have $M\leq N$.

\begin{example}\label{ex:standard_blocks}
    Simple but non-trivial disordered systems can be obtained by sampling from a set of just two blocks $B_1$ and $B_2$. A canonical example of this, which we will consider extensively in this work, is where $B_1=B_{single}$ denotes a single resonator block with $\len(B_{single}) = 1$ and $\ell_1(B_{single}) = s_1(B_{single})  = 2$ while $B_2=B_{dimer}$ is a dimer resonator block with $\len(B_{dimer}) = 2$ and $\ell_1(B_{dimer}) = \ell_2(B_{dimer})=1$ and $s_1(B_{dimer})=1, s_2(B_{dimer})=2$. We choose all wave speeds to be equal to $1$.
    An example of a single realisation  of this system with $M=14$ is described by the following sequence:
    \[
        \chi = (1,1,1,1,2,1,1,1,1,1,2,1,1,1) \in \{1,2\}^M . 
    \]
    In \cref{fig:disorderedsketch}, another realisation corresponding to the sequence $\chi = (1,2,1)$ is illustrated.
\end{example}

Notably, when repeated periodically, both the single resonator and the dimer of resonators share a lower band. However, as will be shown in \cref{fig:prop mat and cdf}, the dimer possesses an additional upper band and the distinct properties of these building blocks affect the resonant frequencies of the resulting array. 

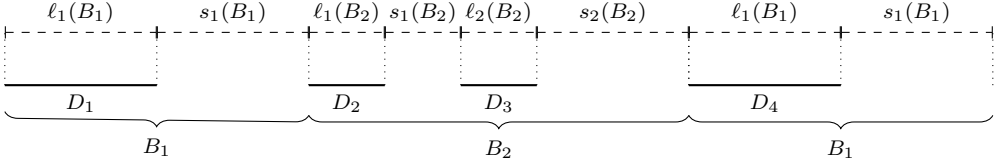
\begin{figure}
    \centering
    \begin{tikzpicture}[scale=1.0, every node/.style={font=\footnotesize}]

    % =======================================================================
    % SINGLE block #1 (length = 2)
    % =======================================================================
    % Resonator from x=0 to x=2
    \draw[thick] (0,0) -- (2,0);
    \node[below] at (1,0) {$D_1$};
    
    % Brace for S block
    \draw[
    decorate,
    decoration={brace, mirror, amplitude=5pt}
    ] (0,-0.35) -- (4,-0.4)
    node[midway, yshift=-0.45cm] {$B_1$};
    
    % Dashed length marking for l
    \draw[-, dotted] (0,0) -- (0,0.7);
    \draw[-, dotted] (2,0) -- (2,0.7);
    \draw[|-|, dashed] (0,0.7) -- node[above]{$\ell_1(B_1)$} (2,0.7);
    
    % Dashed length marking for s
    \draw[-, dotted] (4,0) -- (4,0.7);
    \draw[|-|, dashed] (2,0.7) -- node[above]{$s_1(B_1)$} (4,0.7);
    
    % =======================================================================
    % DIMER block (two resonators, each length = 1)
    % =======================================================================
    % First resonator of the dimer: from x=2 to x=3
    \draw[thick] (4,0) -- (5,0);
    \node[below] at (4.5,0) {$D_2$};
    % Second resonator of the dimer: from x=3 to x=4
    \draw[thick] (6,0) -- (7,0);
    \node[below] at (6.5,0) {$D_3$};
    
    % Brace for D block
    \draw[
    decorate,
    decoration={brace, mirror, amplitude=5pt}
    ] (4,-0.4) -- (9,-0.4)
    node[midway, yshift=-0.45cm] {$B_2$};
    
    % Dashed length markings for each dimer resonator
    \draw[-, dotted] (5,0) -- (5,0.7);
    \draw[-, dotted] (6,0) -- (6,0.7);
    \draw[-, dotted] (7,0) -- (7,0.7);
    \draw[-, dotted] (9,0) -- (9,0.7);
    \draw[|-|, dashed] (4,0.7) -- node[above]{$\ell_1(B_2)$} (5,0.7);
    \draw[|-|, dashed] (5,0.7) -- node[above]{$s_1(B_2)$} (6,0.7);
    \draw[|-|, dashed] (6,0.7) -- node[above]{$\ell_2(B_2)$} (7,0.7);
    \draw[|-|, dashed] (7,0.7) -- node[above]{$s_2(B_2)$} (9,0.7);

    \begin{scope}[shift={(9,0)}]
    % =======================================================================
    % SINGLE block #1 (length = 2)
    % =======================================================================
    % Resonator from x=0 to x=2
    \draw[thick] (0,0) -- (2,0);
    \node[below] at (1,0) {$D_4$};
    \draw[
    decorate,
    decoration={brace, mirror, amplitude=5pt}
    ] (0,-0.4) -- (4,-0.4)
    node[midway, yshift=-0.45cm] {$B_1$};
    
    % Dashed length marking for l
    \draw[-, dotted] (0,0) -- (0,0.7);
    \draw[-, dotted] (2,0) -- (2,0.7);
    \draw[|-|, dashed] (0,0.7) -- node[above]{$\ell_1(B_1)$} (2,0.7);
    
    % Dashed length marking for s
    \draw[-, dotted] (4,0) -- (4,0.7);
    \draw[|-|, dashed] (2,0.7) -- node[above]{$s_1(B_1)$} (4,0.7);
    \end{scope}
\end{tikzpicture}
    \caption{A block disordered system consisting of two single resonator blocks $B_1$ and a dimer block $B_2$ arranged in a chain given by the sequence $\chi= (1,2,1)$. It thus consists of $M=3$ blocks and $N=4$ resonators $D_1,\dots ,D_4$ in total.}
    \label{fig:disorderedsketch}
\end{figure}
\subsection{Sampling}\label{ssec:sampling}
Apart from the makeup of the individual blocks, the crucial factor determining the behaviour of block disordered systems is how these blocks are sampled. The prototypical sampling is to select the block independently and identically with sampling probabilities $p_1, \dots , p_D$, such that $\sum_{d=1}^D p_d=1$. 

To enable later analysis, we now make this sampling formal for infinite sequences in a probability-theoretic sense.
\begin{definition}[Independently and identically sampled discrete sequences]
    We denote the discrete sequence space $\mathbb{L}_D(\Z) \coloneqq \{1,\dots, D\}^{\Z}$ and construct the probability space of \emph{independently and identically sampled discrete sequences} $(\ldz, \mc F, \prob)$ as follows:

    Let 
    \[
        C(i_1,\dots, i_n, X_1\times\dots \times X_n)\coloneqq \{\chi \in \ldz \mid (\oo{i_1},\dots, \oo{i_n}) \in X_1\times\dots \times X_n\}
    \] be the \emph{cylinder set} for $X_1,\dots, X_n \subset \{1,\dots, D\}$ and $i_1,\dots, i_n\in \Z$. On these cylinder sets we define the probability measure
    \[
        \prob (C(i_1,\dots, i_n, X_1\times\dots \times X_n)) \coloneqq F(X_1)\cdot\dots\cdot F(X_n),
    \]
    where $F(X) = \sum_{j=1}^D p_j \mathbb{1}_j(X)$ is the block probability distribution.
    Extending $\prob$ to the minimal $\sigma$-algebra containing all cylinder sets $\mc F$, we obtain the probability space $(\ldz, \mc F, \prob)$.
\end{definition}

We can extend this definition to finite sequences on $\ldm$ by employing the \emph{truncation projection} $\mc P_M:\ldz \to \ldm$ to push forward the probability space $(\ldz, \mc F, \prob)$ to $(\ldm, \mc P_M^*\mc F, \mc P_M^*\prob)$. 

Independent and identically distributed (\emph{i.i.d.}) sampling will prove to be a simple but rich and flexible base case for the investigation of block disordered systems. Therefore, for the rest of this section as well as for sections \cref{sec:jacobi} and \cref{sec:metaatom}, when we talk about \enquote{random block disordered systems} we will consider systems with sequences sampled in an \emph{i.i.d.} manner. However, our analysis is not only restricted to such systems and in \cref{sec:dependent sampling} we shall discuss block disordered systems with other kinds of sampling.

At this point, we would also like to connect the sampling of blocks to the corresponding sequence of resonators. Consider a block disordered system with blocks $B_1, \dots B_D$ sampled with probability $p_1, \dots, p_D$ yielding an \emph{i.i.d.} block sequence $\chi \in \ldz$. This uniquely determines the sequence of resonators, which we shall encode using tuples from the set
\[
    R \coloneqq\left\{ (d,r) \mid d\in \{1,\dots D\}, r\in \{1,\dots, \len(B_d)\}\right\} \subset \N^2.
\]
The tuple $(d,r)$ denotes the $r$\textsuperscript{th} resonator of the $d$\textsuperscript{th} block. We thus get an injective map 
\begin{align*}
    \Phi: \ldz &\to R^{\Z}\\
    \chi &\mapsto \mu \coloneqq (\dots, (\oo{0}, 1), \dots, (\oo{0}, \len(B_{\oo{0}})), (\oo{1}, 1), \dots)
\end{align*}
which we can use to push forward the \emph{i.i.d.} probability space $(\ldz, \mc F, \prob)$ to $(R^{\Z}, \mc F_R, \prob_R)$ where $\mc F_R \coloneqq \Phi^*\mc F$ and  $\prob_R \coloneqq \Phi^*\prob$. 

From this construction, it is easy to see that any resonator sequence $\mu = \Phi(\chi) \in R^{\Z}$ is a bi-infinite Markov chain with finite state space. In particular, if a resonator tuple $(d, r)$ is not at the end of the block (so $r<\len (B_{d})$), then it is deterministically followed by $(d, r+1)$. Conversely, if  $(d, r)$ is at the end of a block (so $r=\len (B_{d})$) then it is followed by one of $\{(1,1), \dots, (D,1)\}$ with probability $p_1,\dots, p_D$, respectively. From this observation, it follows that this Markov chain is homogeneous and irreducible, as the transition probabilities are independent of location and any state is reachable from any other.

\begin{example}
    Consider a block disordered system with blocks that are either a single resonator $B_1 = B_{single}$ or a dimer block $B_2 = B_{dimer}$, as in \cref{ex:standard_blocks}. We sample the blocks \emph{i.i.d.} with $B_1$ and $B_2$ occurring with probability $p_1$ and $p_2$, respectively, yielding a sequence $\chi \in \mathbb{L}_2(\Z)$. 
    The corresponding resonator sequence $\mu = \Phi(\chi) \subset R^{\Z}$ consists of the tuples $R = \{(1,1), (2,1), (2,2)\}$ encoding the single resonator and the two dimer resonators respectively. The finite state space $R$ then allows us to collect the transition probabilities in the \emph{transition matrix}
    \[
    P =\begin{pmatrix}
        p_1 & p_2 & 0\\
        0 & 0 & 1\\
        p_1 & p_2 & 0
    \end{pmatrix},
    \]
    where $P_{ij}$ is the transition probability from the $i$\textsuperscript{th} state to the $j$\textsuperscript{th} state and we encode the block tuples $(1,1), (2,1), (2,2)$ as states $1,2,3$ respectively.
\end{example}

\subsection{Propagation matrix formalism}
For one-dimensional systems, we can characterise the resonator interactions through a series of \emph{propagation matrices}. This will in turn allow us to obtain many of the propagation properties of the total system from the properties of the constituent propagation matrices. 

For the remainder of this subsection we consider an arbitrary block disordered system $\mc D = \bigcup_{j=1}^M B_{\oo{j}} = \bigcup_{i=1}^ND_i$ consisting of blocks $B_1, \dots B_D$ arranged in a sequence $\chi\in \ldm$ and described by a generalised capacitance matrix $\mc C = VC$. Our aim is to identify gaps in the spectrum $\sigma(\mc C)$. In this context, the \emph{propagation matrices} are defined as follows:
\begin{definition}[Propagation matrices]
    For the $i$\textsuperscript{th} subwavelength resonator with length $\ell_i$ and spacing $s_i$ we define the \emph{propagation matrix} at frequency $\lambda$ to be the 2-by-2 matrix propagating any solution $u(x)$ of subwavelength resonance problem \cref{waveeq} from the left edge of the resonator $x_i^{\iL}$ to the left edge of the following resonator $x_i^{\iL}$. Namely
    \begin{equation}\label{eq:propmatdef}
    \begin{pmatrix}
        u_-(x^{\iL}_{i+1})\\
        u_-'(x^{\iL}_{i+1})
    \end{pmatrix} = \underbrace{\begin{pmatrix}
        1-s_i\frac{\ell_i}{v_i^2}\lambda & s_i\\
        -\frac{\ell_i}{v_i^2}\lambda & 1
    \end{pmatrix}}_{P_{l_i,s_i}(\lambda)\coloneqq} \begin{pmatrix}
        u_-(x^{\iL}_{i})\\
        u'_-(x^{\iL}_{i})
    \end{pmatrix},
\end{equation}
where $u_-(x^{\iL}_{i}) = \lim_{x\uparrow x^{\iL}_i}u(x)$ and $u'_-(x^{\iL}_{i})= \lim_{x\uparrow x^{\iL}_i}u'(x)$ denote the \emph{exterior} values of the eigenmode $u(x)$.

We can then characterise the propagation through the entire system using the \emph{total propagation matrix} $P_{tot}(\lambda) \coloneqq \prod_{i=1}^N P_{\ell_i,s_i}(\lambda)$.

Crucially, the propagation matrix $P_{\ell_i,s_i}(\lambda)$ of the $i$\textsuperscript{th} resonator depends \emph{only} on the properties of that resonator.
As a consequence, the \emph{block propagation matrices} $P_{B_d}(\lambda)$ are well-defined as
\begin{equation}
    P_{B_d}(\lambda) \coloneqq \prod_{k=1}^{\len(B_d)}P_{\ell_k(B_d), s_k(B_d)}(\lambda).
\end{equation}
Similarly, we can decompose the \emph{total propagation matrix} $P_{tot}(\lambda)$ as the product of the block propagation matrices
\begin{equation*}
    P_{tot}(\lambda) = \prod_{j=1}^MP_{B_{\oo{j}}}(\lambda).
\end{equation*}
We further note that because $\det P_{\ell_i,s_i}(\lambda) = 1$ for any $\ell_i, s_i, \lambda$ all products of propagation matrices, including block and total propagation matrices, must also have determinant one. We refer to \cite{disorder} for further details on the derivations of the propagation matrices.
\end{definition}

Having defined the block propagation matrices $P_{B_d}(\lambda)$, we now aim to relate their properties to the spectrum $\sigma(\mc C)$ of $\mc C$. We begin by making rigorous the notions of \emph{pass bands} and \emph{bandgaps} for the constituent blocks. To do so, we investigate the eigenvalues of $P_{B_d}(\lambda)$, which we shall denote by $\xi_1$ and $\xi_2$, labelled such that they are sorted by ascending absolute value: $\abs{\xi_1} \leq \abs{\xi_2}$. Because all propagation matrices have determinant one, we find the following characterisation:
\begin{equation}
    \begin{cases}
        \xi_1 = \overline{\xi_2} \in \mathbb{S}^1\subset \C & \text{if } \abs{\tr P_{B_d}(\lambda)}\leq2, \\
        0 < \abs{\xi_1} < 1 < \abs{\xi_2} & \text{if } \abs{\tr P_{B_d}(\lambda)}>2,
    \end{cases}
\end{equation}
where $\tr$ denotes the trace. 
Hence, we can see that a periodised arrangement of the block $B_d$ supports oscillatory extended modes if $\abs{\tr P_{B_d}(\lambda)}\leq2$ while any frequency for which $\abs{\tr P_{B_d}(\lambda)}>2$ must be exponentially growing or decaying and can thus not be a resonant frequency of the periodised system. 
Thus, it is natural to consider a frequency $\lambda\in \mathbb{R}$ in the \emph{pass band} of the block $B_d$ if $\abs{\tr P_{B_d}(\lambda)}\leq2$ and in the \emph{bandgap} if $\abs{\tr P_{B_d}(\lambda)}>2$. For this block, the growth or decay of modes with frequency $\lambda$ passing through is determined by the larger block propagation matrix eigenvalue $\abs{\xi_2}(\lambda)$ of $P_{B_d}(\lambda)$.

In \cite{disorder},  the following Saxon--Hutner-type theorem \cite{saxonhutner} relating the bandgaps of the total propagation matrix (and thus the total system) to the bandgaps of the constituent blocks is proved.
\begin{theorem}\label{thm:saxonhutner}
    Consider a block disordered system with blocks $B_1, \dots, B_D$, sequence $\chi \in \ldm$ and a corresponding generalised capacitance matrix $\mc C$. If a given frequency $\lambda\in \mathbb{R}$ lies in the bandgap of all constituent blocks, that is,
    
    \[
        \abs{\tr P_{B_d}(\lambda)} > 2 \quad \text{ for all }d=1,\dots, D,
    \]
    then $\lambda$ must also lie in the bandgap of the entire system. Namely, 
    $\lambda\notin \sigma(\mc C)$ and $\abs{\tr P_{tot}(\lambda)} >2.$
\end{theorem}

\begin{figure}
    \centering
    \includegraphics[width=0.9\textwidth]{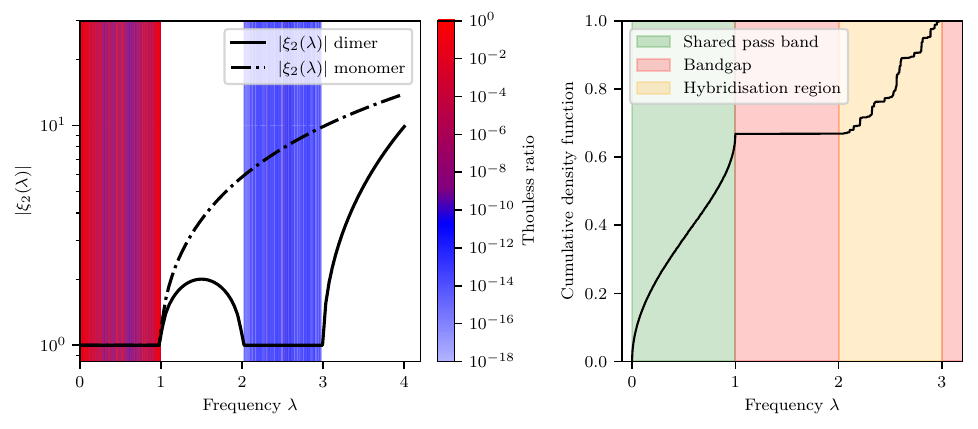}
    \caption{\textbf{Left:} Comparison of maximal block propagation matrix eigenvalue $\abs{\xi_2(\lambda)}$ for both blocks with Thouless ratios $g(\lambda_i)$ of the entire random block disordered system consisting of these block (each colored vertical line corresponds to the Thouless ratio $g(\lambda_i)$ of the eigenvalue $\lambda_i$). Where there are no such lines, the density of eigenvalues is zero. \textbf{Right:} Cumulative density function (CDF) of the total system, laid atop the different spectral regions. We can see that the CDF arrange smoothly in the shared pass band and is constant in the bandgap, reminiscent of the periodic case. However, in contrast to the periodic case, we observe an additional hybridisation region where the CDF grows fractal-like.}
    \label{fig:prop mat and cdf}
\end{figure}

This theorem is illustrated in \cref{fig:prop mat and cdf} for blocks as in \cref{ex:standard_blocks}. We plot the larger eigenvalue $\xi_2(\lambda)$ as a function of the frequency $\lambda\in \mathbb{R}$ for the single resonator and dimer block. Furthermore, we construct a random block disordered system using the respective building blocks and $M=100$ samples, and for each eigenfrequency $\lambda_i\in \sigma(\mc C)$ we also plot the corresponding Thouless ratio $g(\lambda_i)$ as a vertical line. The Thouless ratio is a measure of localisation based on the eigenmode sensitivity to boundary conditions first introduced in \cite{thouless1974Electrons} and adapted to the subwavelength setting in \cite{disorder}. A low Thouless ratio indicates a high degree of localisation for the corresponding eigenmode (see \cref{sec:thouless} for a detailed introduction).

We can see in \cref{fig:prop mat and cdf} that \cref{thm:saxonhutner} accurately predicts the bandgaps of the respective random block disordered systems. 
Moreover, we can see that in the block random setting, the prediction is even stronger. That is, a random block disordered system exhibits a bandgap at $\lambda$ if \emph{and only if} $\lambda$ is in the gap of all constituent blocks, going beyond \cref{thm:saxonhutner}. The \emph{only if} part stems from the fact that, since the blocks are sampled independently, any arbitrarily long sequence of the same block repetitions will occur with probability $1$ as $M\to \infty$. These same-block sequences contribute with eigenvalues wherever their corresponding block has a band, and thus for any frequency $\lambda$ lying in the pass band of any of the constituent blocks, we will find resonant frequencies arbitrarily close to $\lambda$ as $M\to \infty$. 

When the eigenfrequency $\lambda_i$ lies in the pass band of one of the blocks and the bandgap of the others, we can see that the Thouless ratio $g(\lambda_i)$ is much lower compared to the Thouless ratios of the eigenfrequencies lying in the shared pass band. This indicates that the eigenmodes found there exhibit a significantly stronger localisation. Nevertheless, the density of states in these regions is non-zero, distinguishing them from the bandgaps of the total system. Intuitively, each block for which this region does not lie in the bandgap contributes eigenvalues to this region---but because these eigenvalues lie in the gap of some other constituent blocks, they become strongly localised. Since such regions display a mixture of behaviours observed in both the shared pass band and the bandgap (and because of the hybridisation of the bound eigenmodes located there; see \cite{disorder}), we shall call such regions \emph{hybridisation regions}. We obtain the following complete description of the spectral regions of block disordered systems: 
A frequency $\lambda$ lies in the
\begin{description}
    \item[Shared pass band] if and only if it is in the pass band of all constituent blocks;
    \item[Bandgap] if and only if it is in the bandgap of all constituent blocks;
    \item[Hybridisation region] otherwise.
\end{description}
In the bandgap, the density of eigenmodes is zero, and only defect modes can occur (these are exponentially localised modes that are so-called because they are typical when defects are introduced to otherwise periodic or regular systems). Both the shared pass band and hybridisation regions support resonant modes, albeit at strongly distinct degrees of localisation, as demonstrated in \cref{fig:prop mat and cdf}.

This spectral characterisation into three distinct regions can also be observed in the cumulative density function (CDF) of the eigenvalue distribution. \cref{fig:prop mat and cdf} we overlay the CDF over the different spectral regions and observe stark differences: In the bandgap the CDF remains constant, as expected while in the shared pass band the CDF arranges in a smooth curve which is asymptotically vertical towards the band edges---a behaviour mirrors the flat band at the band edges in the periodic case. Finally, in the hybridisation region, we observe behaviour which has no counterpart in the periodic case. The CDF seems to behave fractal-like, staying constant in some parts and performing big jumps in others. The origins of behaviour turn out to be the strong eigenmode decay induced by being in the bandgap of one of the constituent blocks. This insight and the accelerated CDF-estimation method it yields will be further investigated in \cref{sec:metaatom}.

\section{Jacobi operators and metric transitivity}\label{sec:jacobi}
In this section, we will introduce and employ the notion of \emph{metrically transitive operators} from \cite{pastur1992Spectra}. The fact that infinite block disordered systems are indeed metrically transitive Jacobi operators will allow us to obtain a variety of results regarding the spectra of random block disordered systems, such as ergodicity and spectral convergence.

\subsection{Metric transitivity}
An important concept will be the existence of a metrically transitive group on a probability space.
\begin{definition}[Metrically transitive group]
    Let $(\Omega, \mc F, \prob)$ be a probability space and $\mc T$ a topological group of automorphisms on $\Omega$. Then $\mc T$ is called \emph{metrically transitive} if 
    \begin{enumerate}[label=(\roman*)]
        \item $\mc T$ is \emph{stochastically continuous}, \emph{i.e.}, $\lim_{T_1\to T} \prob(T_1X\triangle TX) = 0$ for any $T\in \mc T$ and $X\in \mc F$. Here, $\triangle$ denotes the symmetric difference,
        \item $\mc T$ is measure-preserving, \emph{i.e.}, $\prob(TX) = \prob(X)$ for any $T\in \mc T$ and $X\in \mc F$, 
        \item any $\mc T$-invariant set $X\in \mc F$ has either probability $1$ or $0$.
    \end{enumerate}
\end{definition}

Recall that in \cref{ssec:sampling} we introduced the probability space $(\ldz, \mc F, \prob)$ of block sequences sampled independently and identically and demonstrated that the associated resonator sequences form a probability space of bi-infinite Markov chains $(R^{\Z}, \mc F_R, \prob_R)$. The independence and identical sampling of the blocks allows us to find a \emph{metrically transitive group} $\mc T$ on the space of resonator Markov chains $(R^{\Z}, \mc F_R, \prob_R)$.
\begin{lemma}
    The group of shifts $\mc T = \{T_j \mid j\in \Z\}$ where $(T_j\mu)^{(i)} = \mm{i+j}$ is a metrically transitive group of automorphisms on $(R^{\Z}, \mc F_R, \prob_R)$.
\end{lemma}
\begin{proof}
This result follows from \cite[Example 1.15b)]{pastur1992Spectra} where the essential ingredients are the fact that we consider homogeneous Markov chains and that the finite state space together with irreducibility ensures the uniqueness of the stationary state.
\end{proof}

In the infinite block sequence limit, the subwavelength resonance problem is no longer described by a capacitance matrix $\mc C = VC$ but a Jacobi operator on $\ell_\infty(\Z)$.
\begin{definition}[Jacobi Operator]
    A \emph{Jacobi operator} is a symmetric tridiagonal operator acting on $\lz$. It can be described by two sequences $\bm q^{(i)}, \bm s^{(i)} \in \ell_\infty(\Z)$ and is then given by
    \begin{align*}
        J:\lz &\to \lz\\
        \uu{i} &\mapsto (J\bm u)^{(i)} = -\sss{i}\uu{i-1} + \qq{i}\uu{i} - \sss{i+1}\uu{i+1}.
    \end{align*}
\end{definition}
Because $\qq{i}$ and $\sss{i}$ are bounded and $J$ is tridiagonal and symmetric, $J:\lz\to \lz$ is a well-defined, linear, bounded and self-adjoint operator.

We note that the symmetry requirement on $J$ does not pose any issues since the capacitance matrix eigenvalue problem on $VC$ demonstrated in \cref{thm:capapprox} is equivalent to the symmetric eigenvalue problem on $V^\frac{1}{2}CV^\frac{1}{2}\sim VC$.

\begin{remark}
    As we aim to understand the spectrum of large block disordered systems, we investigate the symmetrised capacitance matrix $\mc J = V^\frac{1}{2}CV^\frac{1}{2}$ in the large resonator array asymptotic $N\to \infty$. Here, $N\to \infty$ is understood as starting with an infinite array of subwavelength resonators (not necessarily arranged in blocks) described by sequences of wave speeds $(v_i)_{i=-\infty}^\infty$, resonator lengths $(\ell_i)_{i=-\infty}^\infty$ and spacings $(s_i)_{i=-\infty}^\infty$ and constructing a sequence of increasingly large capacitance matrices $\mc J_N \in \mathbb{R}^{(2N+1)\times (2N+1)}$ from the truncated sequences $(v_i)_{i=-N}^N, (\ell_i)_{i=-N}^N, (s_i)_{i=-N}^N$. 

    If we consider $(\mc J_N)_{N=1}^\infty$ as a sequence of operators on the sequence space $\ell^2(\Z)$ we find that this sequence converges strongly to a Jacobi operator $J\in \mc L(\ell^2(\Z))$ with the following bands
    \begin{equation}\label{eq:sqformulas}
        \sss{i} = \frac{v_{i-1}v_i}{s_{i-1}\sqrt{\ell_{i-1}\ell_i}} \quad \text{and} \quad 
        \qq{i} = \frac{v_i^2}{\ell_i}\left(\frac{1}{s_{i-1}}+\frac{1}{s_{i}}\right).
    \end{equation}

    Although the spectrum and density of states of this infinite Jacobi operator are generally intractable, we are able to gain much insight about the behaviour of large resonator arrays from its properties, most importantly its ergodicity.
\end{remark}

Given the blocks $B_1, \dots B_D$ any sequence $\chi\in \ldz$ uniquely determines a sequence of resonators $\mu \in R^{\Z}$ and thus also unique sequences of wave speeds $(v_i)_{i=-\infty}^\infty$, resonator lengths $(\ell_i)_{i=-\infty}^\infty$ and spacings $(s_i)_{i=-\infty}^\infty$. Condition \cref{eq:sqformulas} thus allows us to define a Jacobi operator $J(\mu)$ for any resonator sequence $\mu$, making $J:R^{\Z}\to \mc L(\ell_2(\Z))$ a \emph{random operator} on the probability space of resonator sequences $(R^{\Z}, \mc F_R, \prob_R)$.

In this framework, we can now make use of crucial properties of \emph{metric transitivity for operators}. We first need the following definition. 
\begin{definition}
    Let $(\Omega, \mc F, \prob)$ be a probability space with a metrically transitive group $\mc T$. A random operator $A:\Omega \to \mc L(\mc H)$ on a Hilbert space $\mc H$ is called \emph{metrically transitive} if there exists a homeomorphism $T\mapsto U_T$ from $\mc T$ to a group $\mc U = \{U_T\mid T\in \mc T\}$ of unitary operators on $\mc H$ such that 
    \[
        A(T\omega) = U_TA(\omega)U_T^*.
    \]
\end{definition}

The following result holds. 
\begin{proposition}\label{prop:block_is_metrically_transitive}
    Let $(R^{\Z}, \mc F_R, \prob_R)$ be the probability space of resonators obtained from \emph{i.i.d.} block sequences with the metrically invariant group of shifts $\mc T$.
    Associate with each element $d\in \{1,\dots, D\}$ a block $B_d$ containing $\len (B_d)$ resonators characterised by tuples $(v_1,\ell_1, s_1), \dots, (v_{\len(B_d)},\ell_{\len(B_d)}, s_{\len(B_d)})$.  Then, for any $\mu\in R^{\Z}$, the operator $J(\mu)$ constructed following \cref{eq:sqformulas} is metrically transitive with the homeomorphism $T_k \mapsto T_{k}$ mapping a shift on $R^{\Z}$ to a shift on $\ell_2(\Z)$.
\end{proposition}
\begin{proof}
    We begin by noting that the group of shifts $\mc U = \{ T_{k} \mid k\in \Z\}$ is a group of unitary automorphisms as $T_{k}^* = T_{-k}$. Then the fact that 
    \[
        J(T_k\mu) = T_{k}J(\mu)T_{k}^*
    \]
    follows because each entry in $\mu$ directly corresponds to a row of $J(\mu)$ by \cref{eq:sqformulas}. Hence, a shift of $\mu$ by $k$ corresponds to shifting $J(\mu)$ by $k$.
\end{proof}

\subsection{Consequences}
The metric transitivity of $J(\mu)$ ensures its ergodicity and immediately yields a flurry of results. When it is clear from context, we will often write $J$ instead of $J(\mu)$ from now on. First, we find that the spectrum of $J$ and its decomposition are deterministic. 
\begin{theorem}[{\cite[Theorem 2.16]{pastur1992Spectra}}]
Let $J$ be as in \cref{prop:block_is_metrically_transitive}. Then, $\sigma(J), \overline{\sigma_p(J)}, \sigma_{ac}(J)$ and $\sigma_{sc}(J)$ are non-random sets.
\end{theorem}

Since $J$ is self-adjoint for all $\mu\in \ldz$, we can write
\[
    J = \int_{\mathbb{R}}\lambda E_J(d\lambda) ,
\]
where $E_J(\lambda)$ is the \emph{resolution of the identity} w.r.t. $J$, an appropriately chosen family of projection operators non-decreasing in $\lambda$. 
The metric transitivity of $J$ is reflected in the fact that $E_J$ is a metrically transitive projection for any $\lambda$.

\begin{figure}
    \centering
    \includegraphics[width=0.5\linewidth]{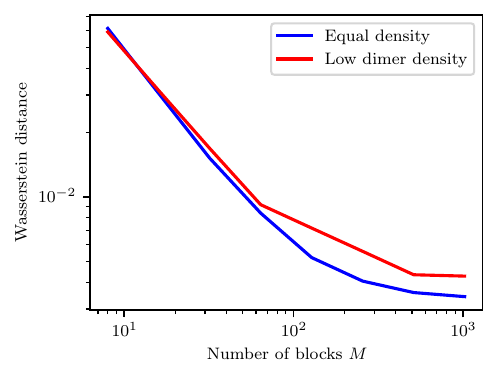}
    \caption{
    Convergence of empirical cumulative density functions under increasing system size.
    We consider random block disordered systems consisting of single resonator and dimer blocks as in \cref{ex:standard_blocks} sampled with either equal density ($p_{single}=p_{dimer}=1/2$) or low dimer density ($p_{dimer}=1/10$). We calculate the eCDF for large arrays ($M=2^{13}$) and compare this to the eCDF of smaller resonator arrays ($M=2^p, p=2,\dots 11$) averaged over $M=10^5$ realisations.}
    \label{fig:ergodicity ws distance}
\end{figure}

\begin{definition}[$K$-truncation]
    We define the $K$-truncation $J_K\in \mathbb{R}^{(2K+1)\times (2K+1)}$ of $J$ to be 
    \[
        (J_K)_{i,j} = J_{i,j} \quad \text{for } \abs{i}, \abs{j} \leq K.
    \]
\end{definition}

For $J_K$, we define the \emph{discrete integrated density of states} by
\begin{equation}
    N_K(J_K, \lambda) \coloneqq \frac{1}{2K+1}\abs{\sigma(J_K) \cap (-\infty,\lambda]},
\end{equation}
\emph{i.e.}, the number of eigenvalues of $J_K$ less or equal to $\lambda$, normalised by the total amount of eigenvalues.

Note that the discrete integrated density of states closely resembles the \emph{empirical cumulative density function} (eCDF) of a finite resonator array of size $N$ described by a capacitance matrix $\mc C\in \mathbb{R}^{N\times N}$ defined as
\begin{equation}
    \eCDF (\mc C, \lambda) \coloneqq \frac{1}{N}\abs{\sigma(\mc C) \cap (-\infty,\lambda]},
\end{equation}
the only difference being edge effects incurred due to truncation, which become negligible as the truncation size $K$ (or equivalently the system size $N$) goes to infinity.

We now have the following ergodicity result on the convergence of the density of states for block disordered systems, which describes how it converges to a deterministic limit
\begin{theorem}[{\cite[Theorems 3.2 and 4.5]{pastur1992Spectra}}]\label{thm:ergodic_convergence}
    Let $J$ be as in \cref{prop:block_is_metrically_transitive}. Then, there exists a non-random positive measure $N(J,d\lambda)$ such that,  with probability $1$, 
    \begin{enumerate}[label=(\roman*)]
        \item $N_K(J_K,d\lambda) \to  N(J,d\lambda)$ weakly as $K\to \infty$;
        \item $N(J,d\lambda) = \mathbb{E}\{\ip{e_0}{E_J(d\lambda)e_0}\}$; 
        \item $\lambda \mapsto N(J,\lambda)$ is continuous. 
    \end{enumerate}
\end{theorem}
We will call $N(J,d\lambda)$ the \emph{density of states} of the Jacobi operator $J$.

This result is demonstrated in \cref{fig:ergodicity ws distance}. In both cases (equal probability and low dimer probability, respectively), the empirical cumulative density functions converge as the system size is increased. The metric used to compare these distributions is the \emph{Wasserstein distance} defined as 
\[
    l_1(U,V) \coloneqq \int_{-\infty}^{\infty}\abs{U(\lambda) - V(\lambda)}d\lambda,
\]
for two cumulative density functions $U$ and $V$.

\section{Fractal density of states in the hybridisation region}\label{sec:metaatom}
Having proven in the last section that the density of states of block disordered systems converges to a deterministic limit, we now aim to find an efficient way of calculating it. We will exploit the fact that, similar to the behaviour shown in \cref{fig:prop mat and cdf}, the density of eigenmodes of disordered systems can be divided into three types of regions: shared pass band where the density of states arranges in a smooth band-like manner, bandgaps where the density is zero and hybridisation regions which exhibit a fractal-like density. 

By fractal-like we mean that, as can be seen in \cref{fig:dos for iid}, the density of states is highly concentrated about a number of self-similar peaks. However, weak hybridisation slightly \enquote{smooths out} this density with increasing system size. This smoothing-out is reflected in the fact that, by \cref{thm:ergodic_convergence}, the integrated density of states of the infinite Jacobi operator is continuous.
Our aim is to better understand the density of eigenfrequencies and in particular their fractal-like arrangement in these hybridisation regions as well as the localisation behaviour of the corresponding eigenmodes. 

Recall that the hybridisation region is characterised by the gap of one (but not all) of the constituent blocks. We will exploit the mode decay this induces to find highly efficient ways of estimating the density of states in these regions and uncover the origin of the fractal-like arrangement of modes.

For the rest of this section, we will limit our analysis and numerics to systems consisting of the two standard blocks (a single resonator and a dimer of resonators) characterised in \cref{ex:standard_blocks}, sampled identically and independently (\emph{i.i.d.}). Its spectral structure can be seen in \cref{fig:prop mat and cdf} and we will focus on its hybridisation region contained in $(2,3)$ lying in the bandgap of the single resonator but the pass band of the dimer block.

\subsection{Perturbative characterisation}
We begin by introducing the main tool which will allow is to quickly construct the density of states. We can view block disordered systems with blocks as in \cref{ex:standard_blocks} as materials of single resonators with an arbitrary and random amount of dimer block defects. Because the upper dimer modes are in the bandgap of the single resonators, these modes are somewhat \enquote{isolated}, allowing us to perturbatively characterise the resulting density of states.

The following will put this intuition on analytical footing:
\begin{proposition}
    Let $A$ be a normal matrix with an eigenpair $(\lambda,\bm u)$ such that $\norm{\bm u}_2=1$. Furthermore, let $\Tilde{A} = A + B$ be a normal perturbation of $A$. Then,
    \begin{enumerate}[label=(\roman*)]
        \item $(\lambda,\bm u)$ is a $\norm{B\bm u}$-pseudoeigenpair of $\Tilde{A}$;
        \item there exists an eigenvalue $\Tilde{\lambda}\in \sigma(\Tilde{A})$ such that $\abs{\Tilde{\lambda}-\lambda}\leq \norm{B\bm u}$.
    \end{enumerate}
\end{proposition}
\begin{proof}
    (i) follows immediately from the definition of $\Tilde{A}$ and (ii) follows from the fact that $\Tilde{A}$ remains a normal matrix.
\end{proof}

This result is useful because it essentially states that the eigenmodes, which are \enquote{small} in some regions, are weakly affected by perturbations in these regions. To estimate the location of the upper eigenfrequencies contributed by dimer blocks, we begin with a dimer block surrounded by a large amount of single resonator blocks in both directions. 
Let $A$ denote the capacitance matrix of this system and $(\lambda, \bm u)$ the dimer mode eigenpair with eigenfrequency in the upper band, namely $\lambda\in (2,3)$.
Because this upper eigenfrequency $\lambda$ lies in the spectral gap of the single resonators, we know its eigenmode decays exponentially when passing through the single resonator blocks (with a decay rate of $\abs{\xi_2(\lambda)}$). This causes the eigenmode to be highly localised around the dimer block.

If we then instead consider an arbitrarily sampled block disordered system $\Tilde{A}$ containing multiple dimer blocks (assuming it contains a dimer in location as before, and there is some amount of single resonators separating the original dimer from the additional dimers) 
and denote $B = \Tilde{A}-A$ we know that $\norm{B\bm u}$ must be small because the exponential decay of $\bm u$ causes it to be small where $B$ is non-zero.
Consequently, the dimer eigenpair $(\lambda, \bm u)$ of the single dimer defect system $A$ will induce a closely corresponding eigenpair in $\Tilde{A}$.

It is also this weak perturbation from $A$ to $\Tilde{A}$ which explains the smoothing-out of the density of states as well as the notion of hybridised bound states. As we move to $\Tilde{A}$, other dimer defect sites and the original defect mode hybridises with new defect modes, causing the defect frequency to perturb slightly.

If, instead of considering only single dimer defects, we also consider local arrangements consisting of multiple dimers located closely together, then we find very close agreement with the peaks of the actual density of states, as can be seen in \cref{fig:dos for iid}.
\begin{figure}[h]
    \centering
    \includegraphics[width=0.95\textwidth]{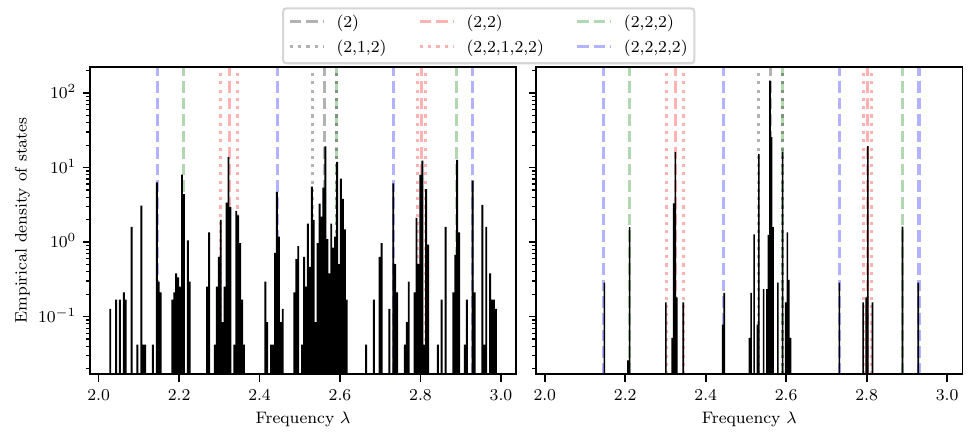}
    \caption{
    Densities of states for block disordered systems ($M = 10^5$) with varying dimer density, together with various dimer defect modes. 
    We can see that in both cases (but especially in the low dimer density case) the peaks in the density of states closely correspond to the defect modes obtained by considering various local arrangements of dimers.
    Both systems consist of single resonator and dimer blocks as described in \cref{ex:standard_blocks}.
    \textbf{Left:} Random sampling with equal probability $p_{single}=p_{dimer} = \frac{1}{2}$.
    \textbf{Right:} Random sampling with lower dimer density $p_{dimer}=\frac{1}{10}$.
    }
    \label{fig:dos for iid}
\end{figure}

This view also explains the self-similarity of the density of states, which is especially visible in the low dimer density setting of \cref{fig:dos for iid}(Right). We observe density peaks at $(2), (2,2), (2,2,2)$ and so on, which seem to get $10$-times smaller with every increase of dimer defect length. The origin of this is that in the low dimer density case with $p_{dimer} = 1/10$, one is, for example, roughly $10$ times more likely to find an arrangement $(2,2)$ than an arrangement $(2,2,2)$ at any point in the block sequence. Hence, the density of states looks somewhat self-similar with increasing dimer sequence length, although this effect is moderated by eigenmode interaction and hybridisation, preventing a truly fractal density.

\subsection{Meta-atom construction}
As seen in the previous subsection, the peaks in the density of states of large block disordered systems are accurately predicted by tracking the defect modes induced by local arrangements of dimers and single resonators. We will call such local arrangements \enquote{meta-atoms}. We can use this to calculate the density of states quickly and accurately by enumerating all relevant meta-atoms, calculating their induced defect eigenmodes and weighing by their relative likelihood.

We will begin by enumerating the meta-atoms. Recall that $\ldz$ denotes the set of sequences on the integers taking values in $\{1,\dots, D\}$, \emph{i.e.}, $\{1,\dots, D\}^{\Z}$, and $\mathbb{L}_D(M)$ its finite counterpart of length $M$, \emph{i.e.}, $\{1,\dots, D\}^M$. Any sequence $\zeta\in \mathbb{L}_D(K)$ will now denote a \emph{meta-atom of length} $K$. Every meta-atom determines through its sequence a local arrangement of single and dimer resonator blocks (in this case, we have $D=2$). We will constrain the enumerated meta-atoms by three factors:
\begin{enumerate}[label=(\roman*)]
    \item the maximal meta-atom length $L$;
    \item the maximal amount of single resonators $P$ 
    contained in the meta-atoms;
    \item every meta-atom must begin and end with a dimer block.
\end{enumerate}

These constraints are motivated as follows. Limiting the length of the considered meta-atoms in (i) ensures that we only consider meta-atoms which occur with high probability; limiting the amount of contained single resonators in (ii) ensures that we only consider meta-atoms which really act as one unit, if the number of contained single resonators increases, the contained dimer blocks begin to increasingly act as smaller individual meta-atoms making it superfluous to consider the present meta-atom; enforcing that meta-atoms begin and end with dimer blocks in (iii) ensures that they do not contain superfluous strings of single resonators on the ends.

We denote the set of all meta-atom sequences satisfying these constraints $\mc M_L^P$,
\begin{equation}
    \mc M_L^P \coloneqq  \bigg\{\zeta \in \bigcup_{k=1}^L\mathbb{L}_2(k) \mid \abs{\{j \mid \zeta^{(j)}=1\}} \leq P,\;\zeta^{(1)} = \zeta^{(-1)} = 2 \bigg\}.
\end{equation}
Condition (i) ensures that $\mc M_L^P$ is finite, namely $\abs{\mc M_L^P} \leq 2^{L+1}$ while conditions (ii) and (iii) limit its growth. In particular, by iterating over the lengths and contained single resonators in $\mc M_L^P$, we obtain
\[
    \abs{\mc M_L^P} = 1+\sum_{p=0}^P \sum_{l=p}^{L-2} \binom{l}{p} = \sum_{p=0}^{P+1}\binom{L-1}{p}.
\]

Having enumerated the meta-atoms, we now aim to calculate their induced defect modes. Let $\zeta\in \ldz$ be a meta-atom. We then attach a set amount $R$ of single resonators on both sides. Namely, we consider $\zeta' = (1)^R\circ \zeta \circ (1)^R$, where $(1)^R = (1,\dots ,1)\in \mathbb{L}_2(R)$ and \enquote{$\circ$} denotes sequence concatenation. This modified sequence now again describes a block disordered system, and we can use the generalised capacitance matrix formulation described in \cref{sec:subwavelength setting} to quickly find the corresponding spectrum $\sigma(\zeta')$. Since we are only interested in the upper eigenmodes, we take the intersection with the interval $(2,3)$ and denote $\mc S_R(\zeta) = \sigma(\zeta') \cap (2,3)$ as the \emph{set of defect modes} associated with the meta-atom $\zeta$. Due to exponential decay, if we consider the meta-atom defect modes in the case of single resonators, a small number of attached single resonators $R\approx 4$ is sufficient in practice.

This already enables us to quickly calculate the upper spectrum of a given block disordered system with sequence $\chi \in \mathbb{L}_2(M)$. As described in \cref{alg:upper_spectrum}, we can estimate the upper spectrum $\sigma(\chi)\cap (2,3)$ by passing through the sequence $\chi$ and repeatedly matching the largest possible meta-atom $\Tilde{\zeta}$ and aggregating the corresponding dimer modes $\mc S_R(\Tilde{\zeta})$ in the multiset $\sigma_{\text{calc}}(\chi)$. We refer to our openly available code for the details concerning the meta-atom calculation.

\begin{algorithm}
\caption{Estimate upper spectrum of block disordered system given block sequence $\chi$}
\label{alg:upper_spectrum}
\begin{algorithmic}[1]

\Require Sequence $\chi \in \mathbb{L}_2(M)$ 
\Ensure Calculated upper spectrum $\sigma_{\text{calc}}(\chi)$

\State $\sigma_{\text{calc}}(\chi) \gets \{\}$ \Comment{Initialise as empty multiset (with repeated elements)}
\State $j \gets 1$

\While{$j \leq \text{len}(\chi)$}
    \State $\Tilde{\zeta} \gets \text{Longest possible meta-atom in } \mc M_L^P \text{ matching } \chi \text{ starting from index } j$
    \State $\sigma_{\text{calc}} \gets \sigma_{\text{calc}} \cup S_R(\Tilde{\zeta})$
    \State $j \gets j + \text{len}(\Tilde{\zeta})$
\EndWhile

\State \Return $\sigma_{\text{calc}}(\chi)$

\end{algorithmic}
\end{algorithm}

In \cref{fig:ecdf_metaatom_calculation}, we can see that, already for fairly low meta-atom lengths $L$ and single resonator count $P$, the cumulative density function is reconstructed very well. Furthermore, we can see that the Wasserstein distance comparing the estimated to the actual density function decreases polynomially as the meta-atom length $L$ increases (we scale $P\sim L/2$ alongside with it).

\begin{figure}[h]
    \centering
    \begin{subfigure}[t]{0.49\textwidth}
        \centering
        \includegraphics[width=\textwidth]{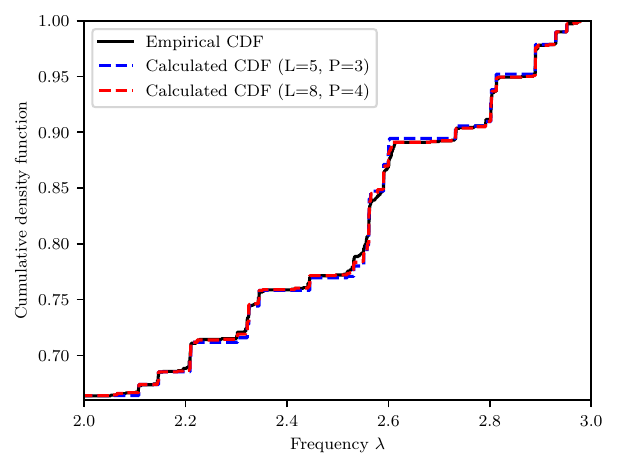}
       \caption{Empirical cumulative density function (CDF) of a block disordered system of $M=10^4$ resonators with equal block probabilities ($p_{dimer}=p_{single} = 1/2$) along with the densities estimated by the meta-atom approach.}\end{subfigure}
    \hfill
    \begin{subfigure}[t]{0.49\textwidth}
        \centering
        \includegraphics[width=\textwidth]{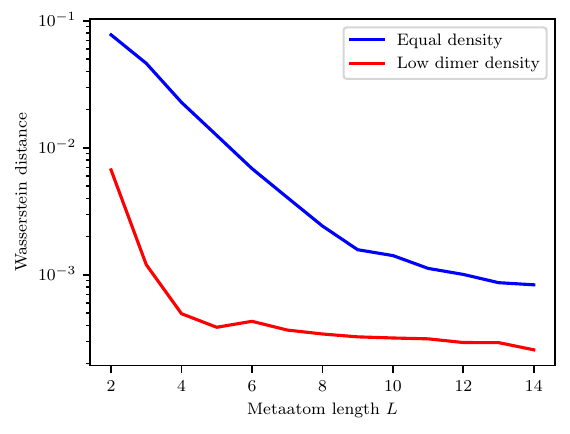}
       \caption{Wasserstein distance comparing the empirical cumulative density function and the meta-atom estimate for increasing meta-atom length $L$ (the amount of permissible single resonators is scaled as $L/2$)}\end{subfigure}
    \caption{Convergence of the meta-atom estimate using \cref{alg:upper_spectrum} to the empirical CDF of a large system as the meta-atom length $L$ and single resonator amount $P$ is increased. Already for small numbers of $L$ and $P$, the estimate agrees extremely well with the empirical CDF.}
    \label{fig:ecdf_metaatom_calculation}
\end{figure}

The great advantage of \cref{alg:upper_spectrum} is its time complexity in the number of resonators $M$. That is, the time required to estimate the spectrum grows linearly in the number of blocks, $\mc O (M)$. While the number of possible meta-atoms grows exponentially as $L$ and $P$ increase, the meta-atom defect modes only have to be calculated once and can be looked up efficiently, yielding a total amortised complexity of $\mc O (ML)$ if $P$ is scaled proportionally to $L$ and $\mc \BO(M\log L)$ if $P$ is held constant.

Finally, we aim to reconstruct the density of states not for a particular block disordered realisation but in expectation. One possible approach to that end is to estimate the relative likelihood of the meta-atoms in $\mc M_L^P$. To do so, we might employ a combinatoric approach or try to estimate it empirically. However, there is also a very straightforward way to use \cref{alg:upper_spectrum} in order to estimate the expected density of states. The key insight is to leverage the ergodicity of our setting, which, following \cref{thm:ergodic_convergence}, ensures that as $M\to \infty$ the finite empirical densities of state converge to a deterministic infinite limit. Essentially, for large numbers of sampled blocks $M$, the ergodicity ensures that the actual density of states varies only very little between the realisations. Therefore, we can estimate both the deterministic density of states in the infinite setting and the expected density of states for a large but finite system simply by sampling $\chi \in \mathbb{L}_D(M)$ for some large $M$ and apply \cref{alg:upper_spectrum} to estimate the density of states. Doing so is easy because the complexity of \cref{alg:upper_spectrum} only scales linearly with $M$.

One issue with this approach and the meta-atom approach in general is that it does not take edge effects into account, which in the case of small $M$ can have a significant impact on the spectral distribution. However, since for small $M$ it is easy to simply calculate the eigenvalues directly, this is not a severe limitation.

\subsection{Variable decay}
\begin{figure}[h]
    \centering
    \begin{subfigure}[t]{0.49\textwidth}
        \centering
        \includegraphics[width=\textwidth]{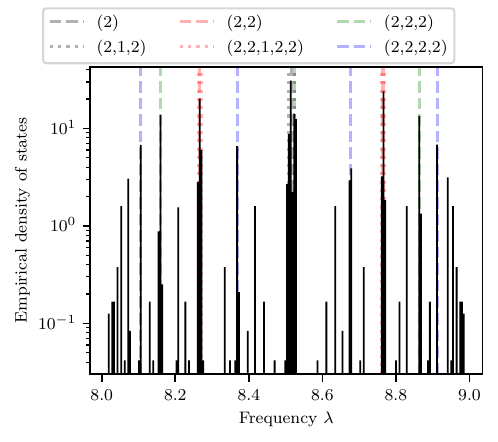}
       \caption{Density of states in the hybridisation region $(8,9)$ of a block disordered system with the standard single resonator block and modified dimer blocks (first dimer spacing $s_1=1/4$), together with various dimer defect modes. $M=10^4$ total blocks sampled with equal density.}\end{subfigure}
    \hfill
    \begin{subfigure}[t]{0.49\textwidth}
        \centering
        \includegraphics[width=\textwidth]{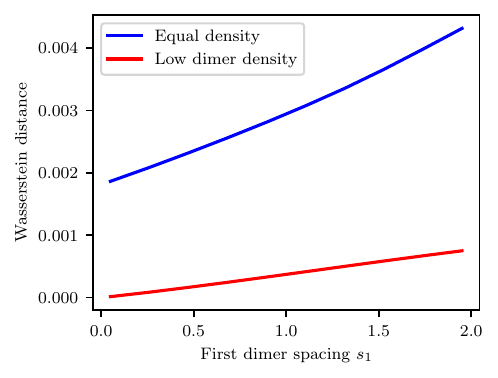}
       \caption{Wasserstein distance comparing the eCDF of a random block disordered system with the respective meta-atom estimate, as a function of first dimer spacing $s_1$.
       }\end{subfigure} 
    \caption{As we decrease the first spacing $s_1$ of the dimer blocks, the corresponding upper band at $(2/s_1, 2/s_1+1)$ gets pushed further into the bandgap $(1,\infty)$ of the single resonator block. This causes weaker hybridisation which in turn increases the accuracy of the meta-atom estimate, as can be seen from the decreasing Wasserstein distance.
    }
    \label{fig:density of states deeper gap}
\end{figure}
As discussed at the beginning of this section, the main ingredient for the meta-atom approximation is that the hybridisation region lies in the bandgap of the single resonator block, causing decay of eigenmodes characterised by the larger eigenvalue $\abs{\xi_2(\lambda)} > 1$ of the single resonator block propagation matrix $P_{B_{single}}(\lambda)$. Thus, it is natural to expect the accuracy of the meta-atom approximation to increase as we get deeper into the bandgap of the single resonator and the larger eigenvalue of the single resonator block $\abs{\xi_2(\lambda)}$ increases. To investigate this, we consider a modified dimer block with $\ell_1 = \ell_2 = 1$ and $s_1 = 1/4, s_2 = 2$. This dimer block has a pass band at $(8,9)$ causing a block disordered system consisting of single resonator blocks and such modified dimer blocks to exhibit a hybridisation region at $(8,9)$ that is significantly deeper in the bandgap $(0,\infty)$ of the single resonator block. As a consequence, the larger eigenvalue of the single resonator block in the middle of new hybridisation region $(8,9)$ is $\abs{\xi_2(8.5)} \approx 32$ compared to $\abs{\xi_2(2.5)} \approx 8$ in the original hybridisation region $(2,3)$.

In \cref{fig:density of states deeper gap}(\textsc{a}), we plot the empirical density of states for the block disordered system with a modified dimer block ($s_1=1/4$). We observe that, compared to the block disordered system with the standard block as in \cref{fig:dos for iid}(Left), the higher decay leaves the meta-atom probabilities unchanged, but causes the peaks to be much more closely spaced together. This in turn causes the meta-atom estimate to be more accurate, as can be seen in \cref{fig:density of states deeper gap}(\textsc{b}) where the Wasserstein distance between the eCDF of the actual block disordered system and the meta-atom estimate decreases as the first dimer spacing $s_1$ approaches $0$. This is true both under equal density sampling and low dimer density. 

\section{Dependent sampling}\label{sec:dependent sampling}
In this section, we investigate the density of states in the hybridisation regions of block disordered systems when the block sequence $\chi \in \ldm$ is no longer sampled independently. We introduce multiple alternative sampling, including hyperuniform sampling and quasiperiodic sequences, and see that in all of these cases the meta-atom approximation continues to function---often even better than in the \emph{i.i.d.} sampling case. Intuitively, this can be understood in the sense that \emph{i.i.d.} sampling is the least structured and thus has the largest number of meta-atoms that need to be considered for an accurate spectral estimate.
\begin{figure}[h]
    \centering
    \includegraphics[width=0.9\textwidth]{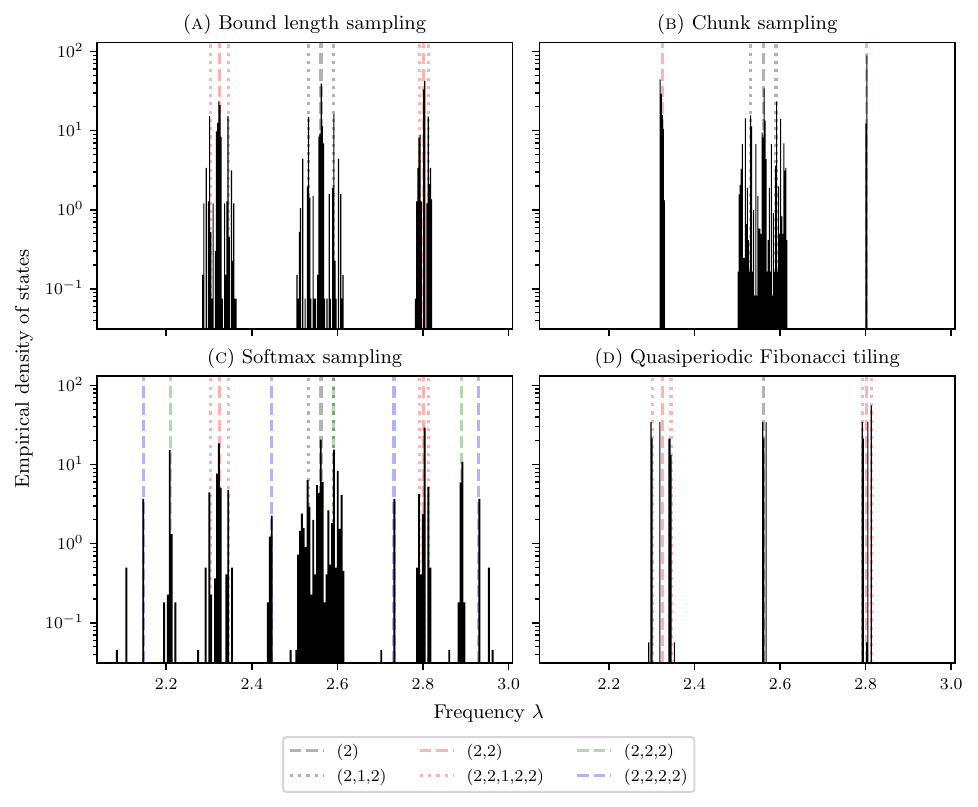}
    \caption{Densities of states for a variety of block disordered systems with dependent sampling, together with various dimer defect modes. All of them consist of single resonator and dimer blocks as described in \cref{ex:standard_blocks}. We can see that the densities are dominated by just a few defect modes.}
    \label{fig:DoS dependent sampling}
\end{figure}

\subsection{Bound length sampling}
The first and simplest type of dependent sampling that we consider is \emph{bound length sampling} whereby for any block sequence $\chi$, a block $d$ may only occur consecutively $B_d$ times. This sampling can be achieved by starting with \emph{i.i.d.} sampled sequences, setting the probability of any sequence that contains more than $B_d$ occurrences of block $d$ consecutively to $0$ and rescaling the remaining probabilities.
Due to the fact that this restriction is homogeneous and does not introduce any long-range correlation, the group of shifts $\mc T$ on $\ldz$ remains a \emph{metrically transitive group}. As a consequence, all the convergence results from \cref{sec:jacobi} continue to hold also under this sampling.

The main difference of bound length sampling compared with \emph{i.i.d.} sampling is the restriction on the possible meta-atoms, as, for example, for a dimer block occurrence length bound $B_2 = 2$, trimer defect modes $(2,2,2)$ no longer appear, as can be seen in \cref{fig:DoS dependent sampling}.
\subsection{Hyperuniform sampling}
An interesting class of disordered systems is the class of \emph{hyperuniform systems} \cite{torquato2018Hyperuniform, torquato_rev,carminati}. These systems are characterised by a mixture of short-range disorder and a suppression of large-scale density fluctuations and thus form an intriguing middle ground between perfectly ordered crystals and completely disordered uniformly random systems.

For a random homogeneous point process $\mc S(x)$ in $\mathbb{R}^d$, hyperuniformity is characterised by the fact that the variance $\Var[\#_R]$ grows slower than the volume over any spherical window where $\#_R$ denotes the number of points within the $0$-centred sphere of radius $R$. Namely, 
\[
    \lim_{R\to \infty} \frac{\Var[\#_R]}{R^d} = 0.
\]
Our aim now is to extend this concept to block disordered systems. A natural approach might be the formalism for hyperuniform two-phase media introduced in \cite{torquato2018Hyperuniform}. Unfortunately, because of the way the resonator blocks are arranged, block disordered systems are not homogenous at short ranges, posing some analytical difficulties. To remedy this, we instead leverage the block disordered structure and extend the concept of hyperuniformity to sequences $\chi \in \ldz$.
\begin{definition}
    Consider a random sequence space $(\ldz, \mc F, \prob)$ such that the group of shifts $\mc T$ on $\ldz$ is a \emph{metrically transitive group}. Furthermore, denote by $\#_{d,R}(\chi)$ the number of occurrences of symbol $d$ in the $R$-truncated sequence $\chi$, namely, 
    \[
        \#_{d,R}(\chi) = \abs{\{j\in \N\cap [-R,R] \mid \chi^{(j)}=d\}}.
    \]

    We then call the random sampling $(\ldz, \mc F, \prob)$ \emph{hyperuniform} if the occurrence count variance $\Var \#_{d,R}(\chi)$ grows \emph{sublinearly} in the truncation window, \emph{i.e.},
    \begin{equation}
        \lim_{R\to \infty} \frac{\Var \#_{d,R}(\chi)}{R} = 0,
    \end{equation}
    for any $\chi \sim (\ldz, \mc F, \prob)$ and $d\in 1,\dots,D$.
\end{definition}

In the following, we shall consider two distinct types of hyperuniform sampling: \emph{Chunk sampling} which severely limits the variety of allowable meta-atoms and \emph{softmax sampling}, which does not.
In the case of chunk sampling, blocks are grouped into short sequences \emph{chunks} such that each chunk contains exactly the same amount of each block. Then, instead of sampling blocks individually and \emph{i.i.d.}, one samples the chunks \emph{i.i.d.}, ensuring that on large scales all the blocks occur exactly equally often. 

In particular, for the two-block case $D=2$, we consider the chunks $\varphi_1 = (2,1)$ and $\varphi_2 = (1,2)$. We can then construct a hyperuniform sequence $\chi\in \mathbb{L}_2(\Z)$ by first sampling an \emph{i.i.d.} sequence $\theta\in \mathbb{L}_2(\Z)$ and then concatenating the chunks accordingly:
\[
    \chi = \dots \circ \varphi_{\theta^{(-1)}} \circ \varphi_{\theta^{(0)}} \circ \varphi_{\theta^{(1)}} \circ \dots.
\]

It is clear that this construction severely limits the variety of available meta-atoms; for example, the meta-atom $(2,2,1,2,2)$ cannot now occur. This can be seen in \cref{fig:DoS dependent sampling}(\textsc{b}) where the density vanishes at the meta-atom $(2,2,1,2,2)$ as opposed to the bound length sampling \cref{fig:DoS dependent sampling}(\textsc{a}) where we see prominent spikes.

Another type of hyperuniform sampling is \emph{softmax sampling}. Here, instead of grouping the sampled blocks in chunks, we sample the blocks individually but adjust the relative block probabilities $p_1, \dots, p_D,$ based on the occurrence count of the previously sampled blocks. We construct the block sequence $\chi\in \ldz$ as follows. We begin by sampling block $\oo{0}$ with probabilities $p^0_1,\dots p_D^0 = p_1, \dots, p_D$, and for any of the following blocks $\oo{\pm(j+1)}$, we modify the probabilities as follows:
\begin{equation}\label{eq:softmax}
    p_d^{\pm(j+1)} = \frac{\exp\left[\beta (p_d(2j+1)-\#_{d,j}(\chi))\right]}{\sum_{d'=1}^D\exp\left[\beta (p_{d'}(2j+1)-\#_{d',j}(\chi))\right]},
\end{equation}
where $\beta\in [0,\infty)$ is called the \emph{temperature} and controls the strength on the occurrence count regularisation. Essentially, when determining the probabilities for the $(j+1)$\textsuperscript{th} and $(-j-1)$\textsuperscript{th} blocks, \cref{eq:softmax} compares the current block counts $\#_{d,j}(\chi)$ with the expected block counts $p_d(2j+1)$ and weighs the sampling probabilities according to the discrepancy. This has the effect that blocks which were sampled more often than expected are less likely to be sampled going forward, and thus represents a kind of regularisation pushing the block occurrence counts closer to their expected values. This in turn effectively means that long-range block occurrence fluctuations are suppressed, making this sampling \emph{hyperuniform} as well.

In contrast to the chunk sampling, softmax sampling does not impose any restrictions on the possible meta-atoms but rather makes meta-atoms consisting of long chains of identical blocks extremely unlikely. This can be seen in \cref{fig:DoS dependent sampling}(\textsc{c}) where already the $4$-dimer meta-atom $(2,2,2,2)$ is significantly less likely as compared to the \emph{i.i.d.} sampled system with the same block illustrated in \cref{fig:dos for iid} (Left).

We refer to \cref{sec:hyper-demo} for a numerical illustration of the hyperuniformity of the two samplings introduced here.

\subsection{Quasiperiodic sampling}
Finally, we consider samplings which are deterministic but not periodic, such as \emph{quasiperiodic sequences} \cite{aperiodic,pastur1992Spectra,simon1982almost}.
The canonical quasiperiodic sequence we consider is the classical \emph{Fibonacci tiling} (see, for instance, \cite{jagannathan2021fibonacci}) that can be constructed using a simple replacement rule.  We define the $j$\textsuperscript{th} \emph{Fibonacci sequence} $\chi_j\in \mathbb{L}_2(M)$ recursively as $\chi_0 = (1)$ and $\chi_{i+1}$ obtained from $\mc \chi_{i}$ using the replacement rules $1 \mapsto 2$ and $2 \mapsto 2, 1$. The spectra of operators with coefficients determined by Fibonacci tilings have been shown to be well approximated by either periodic supercell approximations \cite{damanik2017isospectral, davies.morini2024Super, shubin1978almost} or lifting into a higher-dimensional periodic ``superspace'' (the existence of which is the defining characteristic of quasiperiodicity) \cite{bouchitte2010homogenization, davies2024convergence, rodriguez2008computation}.

A fundamental property of this Fibonacci replacement rule is that $2$ may \emph{at most} appear twice in a row while $1$ may appear at most once in a row. Furthermore, the replacement rules also severely limit the variety of locally occurring meta-atoms, meaning the same patterns appear repeatedly in the structure (albeit, not periodically). This causes the density of states to be dominated by just a small number of meta-atoms, as can be seen in \cref{fig:DoS dependent sampling}(\textsc{d}).
\subsection{Sampling comparison}
\begin{figure}[h]
    \centering
    \includegraphics[width=0.6\textwidth]{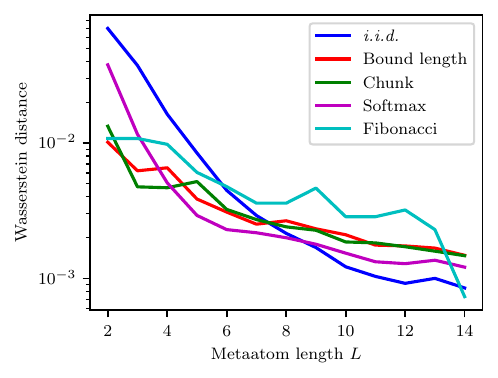}
    \caption{
    Wasserstein distance between the empirical cumulative density function and the meta-atom estimate for increasing meta-atom length $L$ (the amount of permissible single resonators is scaled as $L/2$) for block disordered systems constructed according variety of block sampling methods. The blocks are the standard blocks as in \cref{ex:standard_blocks}. We can see that the accuracy increases with increasing meta-atom lengths $L$, although at different speeds.}
    \label{fig:WS dependent sampling}
\end{figure}
Having introduced all these different types of block sequence samplings, we now aim to compare them with each other and with the \emph{i.i.d.} sampling. Looking at \cref{fig:DoS dependent sampling}, we can see that one of the most visible distinctions is the restriction on the available meta-atoms which some samplings impose. This causes the corresponding densities of state to be concentrated more tightly on the remaining meta-atoms. One would expect that this would work in favour of \cref{alg:upper_spectrum} and allow a more accurate estimation of the density of states compared to the \emph{i.i.d.} sampling case. We investigate this in \cref{fig:WS dependent sampling} where for each of the sampling methods, we sample a block sequence of length $M\approx 15'000$ and calculate the spectrum of the corresponding block disordered system with the standard block as in \cref{ex:standard_blocks}. For increasing meta-atom lengths $L$, we then estimate the spectrum from these block sequences using \cref{alg:upper_spectrum} and plot the Wasserstein distance between the estimated CDF and the empirical CDF. 

For all types of dependent sampling we can see that indeed, for small meta-atom lengths ($L<8$), the spectrum is easier to estimate compared to the \emph{i.i.d.} case and we very quickly obtain a good estimate (Wasserstein distance $<10^{-2}$) already for $L=3$. This confirms the intuition that in the case of dependent sampling, the density is already well captured by a very small amount of meta-atoms. 

However, we also observe that as $L$ increases beyond $8$, a reversal happens and the estimate is more accurate for the \emph{i.i.d.} case. This can be explained by the fact that our choice of the meta-atoms being considered $\mc M^P_L$ is suboptimal in the non-\emph{i.i.d.} case. Namely, it contains many of meta-atoms that cannot occur in the block sequences. At the same time, due to the fact that for these dependent sequences only certain meta-atoms are allowed, a larger part of the density may be concentrated at the high length meta-atoms which \emph{do} occur, requiring a higher meta-atom length $L$ to be captured.

To ameliorate this, the meta-atom space $\mc M^P_L$ might be adjusted depending on the sampling. Either by omitting meta-atoms which do not occur anyway, or by preferentially containing meta-atoms with higher occurrence probability. For a fixed amount $\abs{\mc M^P_L}$ of total meta-atoms under consideration, this might yield a much more efficient density estimation and thus close the approximation gap observed in \cref{fig:WS dependent sampling}.

The one outlier, Fibonacci sampling, demonstrates a possibly counter-intuitive property of \cref{alg:upper_spectrum} for spectral estimation. Namely, as the allowable meta-atom length $L$ is increased, the accuracy of the spectral estimate may get worse. This is due to the fact that \cref{alg:upper_spectrum} greedily matches the meta-atoms to the block sequence and thus may not achieve the overall optimal decomposition of the block sequence into meta-atoms. For random sequences, we can see in \cref{fig:WS dependent sampling} that this does not impose any difficulty as the estimation accuracy almost always increases with increasing meta-atom length $L$. But for the highly structured Fibonacci sequences we can see that the estimation accuracy oscillates in the meta-atom length. Nevertheless, even in the Fibonacci case, the estimation accuracy envelope still increases, and we are also able to obtain a highly accurate estimate of the density of states. This may be surprising considering the complexity of the density of states for quasiperiodic systems.
\section*{Acknowledgments}
   The work of AU was supported by Swiss National Science Foundation grant number 200021--200307.  

\section*{Code availability}
The software used to produce the numerical results in this work is openly available at \\ \href{https://doi.org/10.5281/zenodo.15489499}{https://doi.org/10.5281/zenodo.15489499}.

\appendix
\section{Hyperuniformity characterisation} \label{sec:hyper-demo}
\begin{figure}
    \centering
    \includegraphics[width=0.6\textwidth]{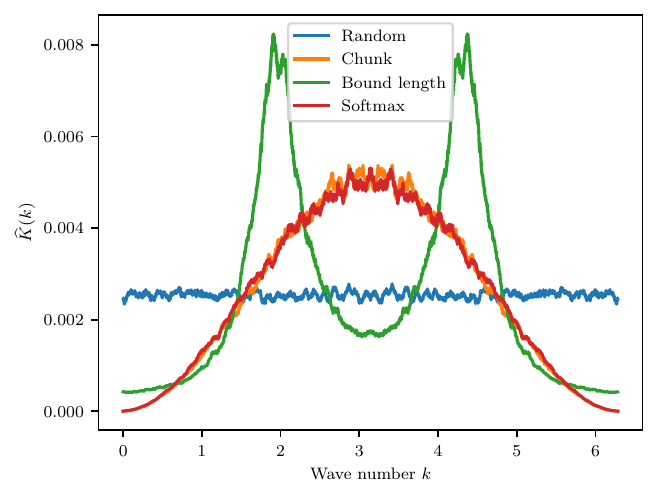}
    \caption{Fourier transform $\hat{K}(k)$ of the autocovariance for a variety of sampling methods. We can see that for both the hyperuniform chunk sampling and the softmax sampling, the Fourier transform $\hat{K}(k)\to 0$ as $\abs{k}\to 0$.
    The autocovariance is calculated empirically over sequences $\chi$ of length $M=10^5$.}
    \label{fig:hyperuniform autocovariance}
\end{figure}
We demonstrate the hyperuniformity of the chunk and softmax samplings in \cref{sec:dependent sampling} by adapting the \emph{autocovariance characterisation of hyperuniformity} for two-phase media found in \cite{torquato2018Hyperuniform} to the discrete two-symbol sequence case. 

For some homogeneous random sequence sampling $(\mathbb{L}_2(\Z), \mc F, \prob)$, we define the \emph{block densities}
\[
    p_d \coloneqq \prob(\{\oo{0}=d\}) \quad d=1,2,
\]
the \emph{two-point probability}
\[
    S_2^{(d)}(r) \coloneqq \prob(\{\oo{0}=\oo{r}=d\}) \quad d=1,2,
\] and the \emph{autocovariance}
\[
 K(r) = S_2^{(1)}(r) - (p_1)^2 = S_2^{(2)}(r) - (p_2)^2
\]
for some $r\in \Z$.

If we assume that $(\mathbb{L}_2(\Z), \mc F, \prob)$ has no \emph{long-range correlation}, then we find that
\[
    K(r) \to 0 \quad \text{as} \quad \abs{r}\to \infty.
\]
We can now characterise the hyperuniformity in terms of the Fourier transform.
Consider 
\[
    \hat{K}(k) = \sum_{r\in \Z} K(r)e^{-\i rk}
\]
periodic in $k\in [0,2\pi]$.
Then, $(\mathbb{L}_2(\Z), \mc F, \prob)$ is \emph{hyperuniform} if and only if
\begin{equation}
    \hat{K}(k) \to 0 \quad \text{as} \quad \abs{k}\to 0.
\end{equation}

In \cref{fig:hyperuniform autocovariance}, we calculate the Fourier transform of the autocovariance for all random samplings introduced in \cref{sec:dependent sampling} as well as the \emph{i.i.d.} sampling. We can see that the chunk and softmax samplings are indeed hyperuniform while the bound length and \emph{i.i.d.} sampling are not, as expected.

\section{Thouless criterion for subwavelength resonators} \label{sec:thouless}
In this section, our aim is to briefly restate the Thouless criterion of localisation for subwavelength resonators, as introduced in \cite{disorder}.

Given some array of $N$ subwavelength resonators described by a generalised capacitance matrix $\mc C\in \mathbb{R}^{N\times N}$, we aim to measure the sensitivity of the $N$ eigenfrequencies $\lambda_1<\dots <\lambda_N$\footnote{Note that $\mc C$ is diagonalisable because it is similar to a Hermitian matrix, and has only simple eigenvalues because it is tridiagonal.} to quasiperiodic boundary conditions. To that end, we take the finite resonator array and periodise it by copying it infinitely often in both directions. As a consequence, we obtain an infinitely periodic array consisting of a unit cell containing $N$ resonators. Using Floquet--Bloch theory, as for example in \cite{ammari.barandun.ea2023Edge}, we find that the spectrum of the periodised system is given by $N$ bands $\lambda_1(\alpha),\dots, \lambda_N(\alpha)$ for some quasimomentum $\alpha$ in the first Brillouin zone $Y^* \coloneqq[-\pi/\iL,\pi/\iL]$. Here, $\iL$ denotes the length of the unit cell, which, by construction, corresponds to the length of the original resonator array.

Following \cite{ammari.barandun.ea2023Edge}, we can consider the \emph{quasiperiodic capacitance matrix} 
\begin{gather}\label{eq:qpcapmat}
    C^\alpha = \left(\begin{array}{cccccc}
         \frac{1}{s_N}+\frac{1}{s_1}& -\frac{1}{s_1}&&&& - \frac{e^{-\i\alpha\iL}}{s_N} \\
         -\frac{1}{s_1}& \frac{1}{s_1}+\frac{1}{s_2}& -\frac{1}{s_2} \\
         & -\frac{1}{s_2} & \frac{1}{s_2}+\frac{1}{s_3}& -\frac{1}{s_3}\\
         &&\ddots&\ddots&\ddots \\
         &&&-\frac{1}{s_{N-2}}& \frac{1}{s_{N-2}}+\frac{1}{s_{N-1}}& -\frac{1}{s_{N-1}}\\
         - \frac{e^{\i\alpha\iL}}{s_N}&&&&-\frac{1}{s_{N-1}}&\frac{1}{s_{N-1}}+\frac{1}{s_N}
    \end{array}\right) \in \mathbb{C}^{N\times N},
\end{gather}
and find the band functions at $\alpha$ to be the eigenvalues of the \emph{generalised quasiperiodic capacitance matrix} $\mc C^\alpha \coloneqq VC^\alpha$. The band functions cannot intersect and the ordering $\lambda_1(\alpha)<\dots < \lambda_N(\alpha)$ holds for any $\alpha\in Y^*$ allowing one to uniquely associate each band function $\alpha\mapsto \lambda_i(\alpha)$ with the corresponding eigenvalue $\lambda_i$ of the finite system, consisting of $N$ resonators.

As can be seen in \cref{eq:qpcapmat}, we need to choose the spacing $s_N$ that corresponds to the separation with which the unit cells are arranged during periodisation. In the block disordered setting, there is a natural choice: 
\begin{equation}\label{eq:sn}
    s_N \coloneqq s_{\len (B_{\oo{M}})} (B_{\oo{M}}).
\end{equation}

We can now estimate the \emph{level shifts} defined as the average shift in the band function frequency
\begin{equation}
    \delta\lambda_i \coloneqq \frac{1}{2\pi\iL}\int_{\alpha\in Y^*} \abs{\lambda_i(\alpha)-\lambda_i(0)}d\alpha .
\end{equation}

In practice, this is well approximated by the periodic and anti-periodic boundary conditions (\emph{i.e.}, $\alpha\in \{0,\pi/\iL\}$), \emph{i.e.},
\begin{equation}\label{eq:energy_shift_approx}
    \delta\lambda_i \approx \abs{\lambda_i(\pi/\iL)-\lambda_i(0)}.
\end{equation}

The \emph{Thouless ratio} $g(\lambda_i)$ for a given frequency $\lambda_i$ is then defined as the ratio of the level shift over the \emph{mean level spacing} $\Delta(\lambda_i)$ around that frequency. Namely,
\begin{equation}\label{eq:thoulessratio}
    g(\lambda_i) \coloneqq \frac{\delta \lambda_i}{\Delta(\lambda_i)} \quad \text{with} \quad \Delta(\lambda_i) = \frac{1}{D(\lambda_i)\iL},
\end{equation}
 where $D(\lambda)$ denotes the \emph{approximated density of states}\footnote{The density of states for our purposes denotes the \emph{expected} number of eigenvalues in the interval $[\lambda - \frac{\mathrm{d}\lambda}{2}, \lambda + \frac{\mathrm{d}\lambda}{2}]$.} at $\lambda$ obtained by interpolation from the empirical eigenvalues $\lambda_1, \dots, \lambda_N$.

In practice, we estimate the mean level spacing $\Delta(\lambda_i)$ empirically by calculating the approximate density of states $D(\lambda)$ using a Gaussian kernel density estimate on the eigenvalues $\lambda_1,\dots, \lambda_N$ of $\mc C$.

Intuitively, the Thouless ratio for a given eigenfrequency $\lambda_i$ as calculated in \cref{eq:thoulessratio} measures the sensitivity of the resonant frequencies to boundary conditions, normalised by the number of eigenfrequencies close by. This normalisation makes sense because, as noted above, the band functions cannot cross and are thus \enquote{sandwiched} between the surrounding bands, limiting their variation proportionally.

Following \cite{thoulessnumerical, thouless1974Electrons}, we relate the Thouless ratio $g(\lambda_i)$ of the eigenfrequencies $\lambda_i$ of $\mc C$ to the localisation behaviour of the associated eigenmodes $\bm u_i$. In particular, we find that
\begin{equation}\label{eq:thoulesscriterion}
    \bm u_i \text{ is } \begin{cases}
        \text{delocalised} & \text{ if } g(\lambda_i) \approx 1,\\
         \text{localised} & \text{ if } g(\lambda_i) \ll 1.
    \end{cases}
\end{equation}
This constitutes the \emph{Thouless criterion of localisation} applied to the disordered systems of subwavelength resonators. 
The core idea is that any localised eigenmode $\bm u_i$ must have exponentially small magnitude $\bm u_i^{(1)}, \bm u_i^{(N)} \approx 0$ at the system edges and thus be insensitive to boundary conditions. This turns out to be a robust way of determining the localisation of eigenmodes, and we refer to \cite{disorder} for a more detailed discussion and demonstration of the relation between the Thouless criterion and localisation.
 \printbibliography

\end{document}